\documentclass[12pt]{article}
\usepackage{amsfonts}
\usepackage{amssymb}
\usepackage{graphicx}
\usepackage{amsmath}
\usepackage{amstext}
\usepackage{bbm}
\usepackage{natbib}
\usepackage{geometry}
\usepackage[onehalfspacing]{setspace}
\usepackage{hyperref}
\usepackage{subcaption}
\usepackage{tikz}

\newcommand{\ubar}[1]{\text{\b{$#1$}}}

\newtheorem{theorem}{Theorem}

\newtheorem{assumption}{Assumption}

\newtheorem{corollary}{Corollary}

\newtheorem{lemma}{Lemma}

\newtheorem{proposition}{Proposition}

\newenvironment{proof}[1][Proof]{\textbf{#1.} }{\ \rule{0.5em}{0.5em}}

\usepackage[]{color-edits}

\addauthor{ni}{blue}
\addauthor{nw}{purple}
\addauthor{bl}{orange}

\setcitestyle{authoryear}

\title{Maximal Procurement under a Budget\thanks{Random author order verifiable with record locator qJFy0RKB-tdM at \\
\texttt{https://www.aeaweb.org/journals/policies/random-author-order/search}. We would like to thank Dirk Bergemann, Peter Crampton, Laura Doval, Tan Gan, Johannes H\"orner, Yingkai Li, Bart Lipman, Michael Ostrovsky, Marzena Rostek, Kai Hao Yang, Jidong Zhou, and participants at the NBER Market Design working group and Yale microeconomic theory lunch seminar and breakfasts for helpful comments and discussions.  A previous version of this paper was titled ``Maximizing the Effect of Altruism.''}}

\author{  Nicole Immorlica\thanks{Microsoft Research} \and \textcircled{r}\and Nicholas Wu \thanks{Department of Economics, Yale University, New Haven, CT 06511,
nick.wu@yale.edu } \and \textcircled{r}\and Brendan Lucier \footnotemark[2]}

\begin{document}
\begin{titlepage}

\maketitle

\begin{abstract}
We study the problem of a principal who wants to influence an agent's observable action, subject to an ex-post budget. The agent has a private type determining their cost function. This paper endogenizes the value of the resource driving incentives, which holds no inherent value but is restricted by finite availability. We characterize the optimal mechanism, showing the emergence of a pooling region where the budget constraint binds for low-cost types. We then introduce a linear value for the transferable resource; as the principal's value increases, the mechanism demands more from agents with binding budget constraint but less from others. 
\end{abstract}

\end{titlepage}

\section{Introduction}


In many natural scenarios, an ex-post budget constraint arises. In many institutional settings, the allocation of funds is separate from their usage; one government body might determine the amount of funding to go towards a certain goal, while a different government agency oversees the usage of those funds. In other situations, the resource being used to incentivize an agent can be non-monetary; with prizes or prestige, the agent naturally has a bounded value for the prize, and the principal's value for the prize is endogenous and determined by the mechanism.


We study the strategic interaction between a principal, endowed with a limited amount of a transferable resource, who attempts to influence the action of an agent with private information about its own costs. The scenario differs from the standard mechanism design literature in two important ways. Firstly, the principal is budget-constrained; that is, the principal cannot spend more than an upper bound regardless of the private information of the firm. Secondly, the value of the resource is endogenous to its ability to influence the agent; the principal has no value for the budget themselves, and only cares about its impact through the agent. This model is sufficiently general to also characterize many related design problems; we provide additional examples in Section \ref{sec:model}.

In this paper, we extend the mechanism design and delegation literature by proposing a model for providing incentives under a budget constraint. We show that the optimal solution is implementable as a subsidy schedule, where the principal compensates the agent based on the observed change in production. We show that the schedule partitions the agent types into more and less efficient types, offering the entire budget for types that meet a threshold and discriminating on less efficient types. 

Intuitively, the primary concern for output induces a fully \textit{endogenous} value of money for the principal. The principal's schedule can be understood as selecting some threshold type and splitting the agent types. All the more efficient agents are pooled; they are required to satisfy a single output threshold and given the entire budget. The more efficient agent types determine the shadow value of money; the less efficient agents face a schedule that optimizes with respect to this shadow value. The threshold trades off a welfare loss from withholding the budget from the less efficient agents with information rents it must pay the more efficient types. 
A positive measure of the most efficient types are always pooled, even as the budget increases.  To illustrate this point starkly, we consider agent cost functions that satisfy a separability condition. We show that under this separability condition, the shape of the schedule is invariant to the size of the budget. In other words, no matter how large the budget is, the threshold for the pooling region never changes.

We extend the model to allow a competitive ``outside option'' for the principal; that is, the unused residual budget has a linear value. 
We show that provided the value of the budget is not high, the same pooling region of the optimal schedule persists. The linear value case highlights the importance of the threshold type; as the principal's value for the resource increases, the principal asks more from types more efficient than the threshold, and asks less from types less efficient than the threshold. We can extend our methodology to a broader class allowing nonlinear, convex budget costs as well. However, the ex post budget constraint feature is critical to the results; the pooling region does not arise if the budget only needs to be satisfied ex ante. 

We also extend the model to allow competition between $N$ agents vying to contract with the principal.  The principal can choose at most one of the agents to contract with based on their reported types.  This competition between agents relaxes the incentive constraints on the principal's schedule, as the selection rule gives agents an additional incentive to declare low costs.  We characterize the optimal transfer schedule for a principal who contracts with the best declared agent type.  We show that, unlike in the single-agent case, the optimal schedule may not be monotone under our ex post budget constraint.

The paper is structured as follows. In the next subsection, we preview the related literature in mechanism design and delegation and distinguish this particular situation from similar work. Section~\ref{sec:model}  details the model and assumption. Section~\ref{sec:general} formalizes the general mechanism design problem and highlights the important features of the model. Section~\ref{sec:linear} discusses the model with a linear value for the resource. Section~\ref{sec:extensions} characterizes other variants of the model including our extension to multiple agents, and Section~\ref{sec:conclusion} concludes.

\subsection{Literature}
Our model is related to the literature on contract design in settings with moral hazard and private actions going back to \cite{gh83} and \cite{h79}. In this line of work, a principal contracts with an agent to take a costly action that is not directly observable, but that results in an observable outcome that is payoff-relevant to the principal and upon which payments can be contingent. 
In contrast to this line of work, in our model the principal can observe the outcome (i.e., production) but has uncertainty over the cost function of the agent. Our work also introduces an ex-post constraint on transfers.

Additionally, our work is related to literature on monopoly regulation. Closely related work to ours is that of \citet{bm82}, which studies a social planner problem for regulating a monopolist, and \cite{lt86}, which introduces noisy cost observations. However, our work is distinct in multiple ways. Firstly, our environment considers a budget-constrained principal, whereas there are no constraints on transfers in the other models in the literature. Additionally, the only objective for the principal in our scenario is to maximize production, where as the principal in the monopoly regulation environment typically maximizes a weighted sum of consumer surplus and producer profits. There are other minor distinctions in our model; we allow for a broad class of cost functions for the agent, whereas the literature typically assumes particular functional forms.

Given that our work involves nonstandard value for transfers, we complement the literature on delegation. The original delegation framework was formulated to \cite{holmstrom1980theory}, which introduced the general class of delegation problems and provides conditions for existence of a solution. \cite{holmstrom1980theory} and follow-up delegation models \citep{am08,ab13,abf18,hy20} typically assume that contingent transfers are unavailable to the principal, and instead the principal retains the right to take a state-dependent action that impacts utilities of both the principal and agent. In contrast to this line of literature, our model does allow for contingent transfers between the principal and agent, but in a special way; the transfers do not factor into the principal objective and face a budget constraint. 

Given that the principal has no intrinsic value for the budget herself, our problem relates to the literature on budget-feasible mechanism design, initiated by \citet{singer2010budget}.  That paper studies a procurement problem where the designer wishes to purchase a utility-maximizing set items from sellers subject to a budget constraint.  A large literature followed, exploring the approximately optimal design under different combinatorial valuation functions for the principal and varying information structures and solution concepts (see, e.g., \cite{amanatidis2019budget,balkanski2022deterministic,bei2012budget,gravin2020optimal,leonardi2021budget} and references therein). Our work is differentiated from this literature in that the principal is maximizing the variable action of a single agent rather than maximizing over a combinatorial set of items, and that agent has a convex cost of procurement. 

Also related is the literature on contest design.  This literature studies a setting in which a designer with a fixed budget wishes to solicit submissions to a contest.  The designer can split the budget among potentially multiple prizes.  The objective can be to optimize the aggregate quality of submissions \citep{archak2009optimal,moldovanu2001optimal}, or the quality of the best submission \citep{chawla2019optimal,che2003optimal,ghosh2016optimal,moldovanu2006contest}.  Contests with multiple agents is reminiscent of the extension of our problem to multiple firms, discussed in Section~\ref{sec:extensions}.  In the contest design literature, agents work to potentially receive a prize from a set of prizes, with fixed utility values. In contrast, our work can be interpreted in the lens of contest design but providing incentives using a single perfectly divisible prize.

Methodologically, our paper uses optimal control techniques and arguments similar to \cite{lt86} and \cite{l03}. Our design solution shares similar qualitative features as the hidden-information case of \cite{l03}, but the problems are different in a substantial way; because the control problem in \cite{l03} arises from a dynamic relationship, the constraint on the agent's action schedule in the design problem is bounded above by the value of the design problem and a term depending on the discount factor, whereas in our problem, the constraint on the agent's action schedule is instead given by an exogenous budget. This difference yields important distinctions in the outcome; in \cite{l03}, the first-best is sometimes possible, whereas it is never possible for a non-degenerate distribution of types in our case. 

\section{Model}
\label{sec:model}
We first lay out the formalism for the model, and then discuss applications at the end of this section. Consider a principal-agent design problem where the principal would like to provide incentives for an agent to take a costly action. The action space is $X = \mathbb{R}_+$; an action $x \in X$ can be interpreted as an amount of a good produced by the agent. Without loss, we normalize the space $X$ so that the principal derives utility $x$ from the agent taking action $x$. The agent has a private cost function $\Psi: X \times \Theta \to \mathbb{R}_+$, which determines the cost the agent incurs from producing $x \in X$ given the agent's type $\theta \in \Theta = [\ubar{\theta}, \bar{\theta}] \subset \mathbb{R}_+$. Without loss, we normalize the cost function $\Psi(0, \theta) = 0$ for all $\theta$.

The agent's type is drawn prior to the start of the game by Nature from a full-support distribution $\mu \in \Delta(\Theta)$. The type is privately observed only by the agent, but the principal has belief $\mu$ over the agent type. We assume the belief $\mu$ admits a probability density function $f$, with corresponding cumulative distribution function $F$.

The principal has a budget $T$ of a transferable resource that the agent values. The agent's utility is quasilinear in the cost and the transfer: that is, if the agent of type $\theta$ produces $x$ and receives a transfer $t$, his utility is $u_A(x, \theta) = t -\Psi(x, \theta)$. To be clear, we suppose the principal can observe the action $x$, as opposed to the moral hazard problem in \cite{h79} and \cite{gh83}; instead, the principal must decide how to award funds given the observed action, without precisely knowing the agent's cost \textit{function}. 

To highlight the key economic forces, we will present the baseline model where the principal has no value for the transferable resource. In Section \ref{sec:linear} we consider allow the principal to have a constant value the resource, and consider a generalized value function in the extensions in Section \ref{sec:extensions}.

\begin{assumption}
We assume the cost function $\Psi$ is nonnegative, twice continuously differentiable, supermodular, and strictly increasing in both arguments. Additionally, we assume that $\Psi$ is convex for all $\theta$ and its convexity $\Psi_{xx}$ is increasing in $\theta$. \footnote{Such assumptions appear in the cost-function considered by the hidden-information in \cite{l03}. However, in \cite{l03}, there is an additional assumption that $\Psi_{\theta \theta x}$ is nonnegative, but we do not require this.}
\end{assumption}

A couple of remarks about the assumptions and game structure are in order. The supermodularity assumption is indespensible, and intuitively requires that higher types suffer steeper costs. The $C^2$ differentiability assumption is technical and allows us to use cross-derivatives. The convexity assumptions ensure validity and sufficiency of first-order optimization conditions, and intuitively require that agent types are ranked in how convex their costs are.
For instance, in an example where the principal offers money to a firm, the assumption requires that the higher $\theta$ firms are less scalable (i.e. their production costs grow faster). In this setting, the type $\theta$ can be interpreted as an inefficiency parameter.

By standard arguments, we invoke the revelation principle, so we consider direct mechanisms where the agent reports his type to the principal subject to incentive compatibility and individual rationality constraints. (We will see later that all feasible direct mechanisms will be implementable by a subsidy schedule.)

More precisely, the mechanism design problem requires the principal to select functions $x: \Theta \to X$ and $t: \Theta \to \mathbb{R}_+$, in the following optimization problem:
\begin{align}
& \max_{x, t} && \int_{\theta \in \Theta} x(\theta) f(\theta) \ d\theta \label{prblm:general} \\
& \textnormal{subject to} && t(\theta) - \Psi(x(\theta),\theta) \ge  t(\theta') - \Psi(x(\theta'),\theta) && \forall \theta, \theta'\in \Theta && \textnormal{(IC)} \notag \\
&&& t(\theta) - \Psi(x(\theta),\theta) \ge 0 && \forall \theta \in \Theta  && \textnormal{(IR)} \notag \\ 
&&& t(\theta) \le T && \forall \theta \in \Theta &&\textnormal{(B)} \notag 
\end{align}

To understand the problem, we provide three example applications of the model. 

\paragraph{Example (Emissions Reduction)} Consider a policymaker trying to reduce emissions from a large population. Individuals in the population have private emissions abatement costs $\theta \sim F$ which determine how severely emissions reductions impact their utility; that is, $\Psi(x, \theta)$ is how costly an individual of type $\theta$ finds it to reduce their emissions by $x$. 
The policymaker cannot pay any single individual more than an upper bound $T$, but would like to induce the largest possible collective emissions reduction. The optimal mechanism can be implemented by a subsidy schedule, which determines how much compensation an individual receives as a function of the individual's emissions reduction.

\paragraph{Example (Altruism)} Suppose an altruist or a government has pledged a budgeted sum $T$ to spur production of a product (vaccine, baby formula) during a shortage. The altruist is the principal, and the monopoly producing the product is the agent, whose private type characterizes their production costs. The action $x$ corresponds to the amount of the product produced, and the resource is the budgeted money. The altruist seeks to maximize production conditional on respecting the budget constraint.

\paragraph{Example (Non-monetary Incentives)} Consider an academic institution interested in increasing its research output by hiring a new researcher. The institution can grant the researcher a non-monetary award (title, prize, tenure) which has no intrinsic value to the institution itself. The researcher values the award at $T$, and has a private cost of producing output $\Psi(\cdot, \theta)$. The institution can assess the output of the researcher and decide with what probability to give the award. The mechanism in this example corresponds to the institution committing to a schedule that maps research output into a probability of receiving the award.

\section{Optimal Mechanisms}
\label{sec:general}
In this section, we first present the mathematical characterization of the optimal mechanism. We then discuss some of the qualitative features of the model, notably the emergence of a pooling region. We provide some economic intuition for the optimal mechanism as splitting the type space into ``efficient'' and ``inefficient'' types and optimizing distinct shadow cost functions for the two. Finally, we show that the optimal mechanism admits a closed-form under a separability assumption on $\Psi$, and in this case, the size of the pooling region is invariant to $T$. 

Solving our mechanism relies on a lemma which characterizes the feasible schedules, or the set of allocations $x: \Theta \to \mathbb{R}_+$ for which there exists a transfer schedule $t: \Theta \to [0, T]$ that satisfies (IC), (IR), and (B). The lemma follows the standard mechanism design techniques to derive the feasible set.

\begin{lemma}
\label{lem:general_feasible}
    $x: \Theta \to X$ is a feasible schedule if and only if the following two conditions hold:
    \begin{enumerate}
        \item $x$ is nonincreasing.
        \item $\Psi(x(\ubar{\theta}),\ubar{\theta}) + \int_{\ubar{\theta}}^{\bar{\theta}} \Psi_\theta(x(s), s) \ ds \le T $
    \end{enumerate}
    Further, given a feasible schedule $x$, a transfer function that supports the schedule is given by 
    \begin{equation}\label{eqn:gen_transfers}
    t(\theta) = \Psi(x(\theta), \theta) +  \int_{\theta}^{\bar{\theta}} \Psi_\theta(x(s), s) \ ds 
    \end{equation}
\end{lemma}

The first condition is a monotonicity constraint, and the second condition constrains the magnitude of $x$. Note that the normalization constraint (2) comes from the transfer bound, the IC constraints, and the envelope theorem, and so the feasible set characterization does not depend on the prior. Note the first part of the Lemma implies that these two conditions are necessary and sufficient for a feasible schedule, and the second part explicitly determines the transfer schedule that implements any feasible schedule.

The proof of Lemma \ref{lem:general_feasible} is standard; we use supermodularity to show that $x$ must be nonincreasing, and then apply the envelope theorem to the IC constraints to derive the transfer schedule. The transfer schedule in Lemma~\ref{lem:general_feasible} is analogous to the Myersonian transfer schedule, since the transfer schedule is also derived from using the envelope theorem to integrate out interim utility of the firm from the incentive compatibility conditions. 

The integral expression in the transfer expression is the information rent that type $\theta$ receives; since $\Psi(x(\theta), \theta)$ is exactly the loss the firm experiences from producing $x(\theta)$, the transfer compensates the firm for the loss and provides additional rent, which is larger for more efficient types (lower $\theta$). 

As a corollary, since $\Psi$ and $\Psi_\theta$ are convex, conditions (1) and (2) jointly determine a strictly convex subset of the space of functions from $\Theta \to \mathbb{R}_+$.
Thus, as a corollary of our lemma, since the mechanism design problem maximizes a linear functional over a strictly convex set, the optimal mechanism is unique. 
\begin{corollary}\label{corr:unique_soln}
    If a solution exists, the solution to the generalized mechanism design problem in (\ref{prblm:general}) is unique.
\end{corollary}

Further, Lemma \ref{lem:general_feasible} implies that the mechanism is implementable by a subsidy schedule or action-contingent transfers; that is, the optimal mechanism $x(\theta), t(\theta)$ parametrizes a well-defined function from $X \to \mathbb{R}_+$. 

\begin{corollary}\label{corr:subsidy_schedule}
    Suppose $\theta < \theta'$, and let $x, t$ be a feasible mechanism. If $x(\theta) = x(\theta')$, then $t(\theta) = t(\theta')$. 
\end{corollary}
\begin{proof}
    By Lemma \ref{lem:general_feasible}, since $x$ must be nonincreasing and $x(\theta') = x(\theta)$, $x$ must be constant on $(\theta, \theta')$. As a minor abuse of notation, let $x$ be that value on $(\theta, \theta')$. Then 
    \begin{align*}
         t(\theta) &= \Psi(x, \theta) + \int_\theta^{\bar{\theta}} \Psi_\theta(x(t), t) \ dt \\
         &= \Psi(x, \theta) + \int_\theta^{\theta'} \Psi_\theta(x, t) \ dt + \int_{\theta'}^{\bar{\theta}} \Psi_\theta(x(t), t) \ dt \\
         &= \Psi_\theta(x, \theta') + \int_{\theta'}^{\bar{\theta}} \Psi_\theta(x(t), t) \ dt = t(\theta')
    \end{align*} 
\end{proof}

Since $x,t$ are both nonincreasing, the corollary implies that any feasible mechanism $x(\cdot), t(\cdot)$ can be implemented as a subsidy schedule $\hat{t}: X \to \mathbb{R}_+$. 

With some more work, we can show the existence of a solution and characterize it with a differential equation system.

\begin{theorem}\label{thm:general_mech}
    An optimal mechanism $(x, t)$ for \eqref{prblm:general} exists and is unique. The optimal mechanism, together with a nonnegative Lagrange multiplier $\lambda$ and a nonnegative, absolutely continuous costate function $\rho$, jointly solve the following system of differential equations with boundary constraints: 
     \begin{gather} 
        \rho(\theta) > 0 \implies \dot{x}(\theta) = 0 \label{cond:comp_slackness} \\
        \dot{\rho}(\theta) =  \lambda \Psi_{x\theta}(x(\theta), \theta) - f(\theta)\label{cond:costate_evol} \\
        \rho(\bar{\theta}) = 0 \label{cond:upper_bound}\\
        \rho(\ubar{\theta}) = \lambda \Psi_x(x(\ubar{\theta}),\ubar{\theta})  \label{cond:lower_bound} \\
        t(\ubar{\theta}) = \Psi(x(\ubar{\theta}),\ubar{\theta}) + \int_{\ubar{\theta}}^{\bar{\theta}} \Psi_\theta(x(s), s) \ ds  = T \label{cond:budget_bind}
    \end{gather}
    where \eqref{cond:costate_evol} holds wherever $\rho$ is differentiable.
\end{theorem}

To understand the result, note that the feasible set characterization in Lemma \ref{lem:general_feasible} allows us to focus solely on characterizing $x$, subject to a nonincreasing condition and a rewritten budget constraint. We handle the budget constraint by writing the Lagrangian, and deal with the nonincreasing condition by taking an optimal control approach\footnote{See \cite{lt86}} where `time' corresponds to the type interval, the control state variable is the required action $x$, and control variable is the derivative of $x$. We can interpret the conditions in Theorem \ref{thm:general_mech} through this control framework. The costate $\rho$ denotes the shadow cost of the monotonicity constraint. The first condition \eqref{cond:comp_slackness} is the complementary slackness condition; wherever the constraint's shadow cost $\rho^*$ is positive, the constraint is binding (i.e. $x$ is constant). The second condition \eqref{cond:costate_evol} dictates the costate evolution. The constraints \eqref{cond:upper_bound} and \eqref{cond:lower_bound} establish boundary conditions. The final constraint \eqref{cond:budget_bind} establishes that the budget binds at the most efficient / least costly type. 

The proof intuition follows three steps. First, we write out the Lagrangian relaxation of the problem using the feasible set characterization from Lemma \ref{lem:general_feasible}. We formulate the optimization as an optimal control problem, where the control variable governs the derivative of $x$ and is constrained to be nonpositive wherever $x$ is differentiable. 

Second, we show that the conditions in Theorem \ref{thm:general_mech} are necessary and sufficient for an optimal mechanism. Fixing $\lambda$, the first four conditions \eqref{cond:comp_slackness}-\eqref{cond:lower_bound} are necessary and sufficient for a solution to the Lagrangian relaxation control problem; necessity follows by invoking the Pontryagin maximum principle, and sufficiency follows from Arrow's sufficiency condition as proved in \cite{ks71}. Sufficiency of $\lambda > 0$ and \eqref{cond:budget_bind} follows from Lagrangian duality; we show necessity by proving no solution to \eqref{cond:comp_slackness}-\eqref{cond:lower_bound} exists for $\lambda \le 0$. 

Finally, we prove existence. We use the theory of ordinary differential equations to show that \eqref{cond:comp_slackness}-\eqref{cond:lower_bound} has a solution for $\lambda > 0$. We then use the intermediate value theorem to show that there must exist a $\lambda$ such that \eqref{cond:budget_bind} holds. Since we proved in the second step that these conditions were necessary and sufficient for a solution, we thus showed that an optimal mechanism exists. Corollary \ref{corr:unique_soln} implies that this solution must also be unique.

To understand how the characterization in Theorem \ref{thm:general_mech} can be used to determine the optimal mechanism, suppose $\lambda$ were exogenously fixed. Consider the illustration in Figure \ref{fig:general_mech_illustration}. Complementary slackness implies that when $\rho$ is positive, $x$ must be constant, and whenever $x$ is decreasing, $\rho$ must be constant at zero, and hence $x$ must follow the trajectory implied by the implicit equation derived from $\dot{\rho} = 0$. The boundary conditions dictate the start and end values of the costate $\rho$. Note that altering $\lambda$ either uniformly increases or decreases $x$; intuitively, at the solution value of the Lagrange multiplier, the budget constraint exactly binds.

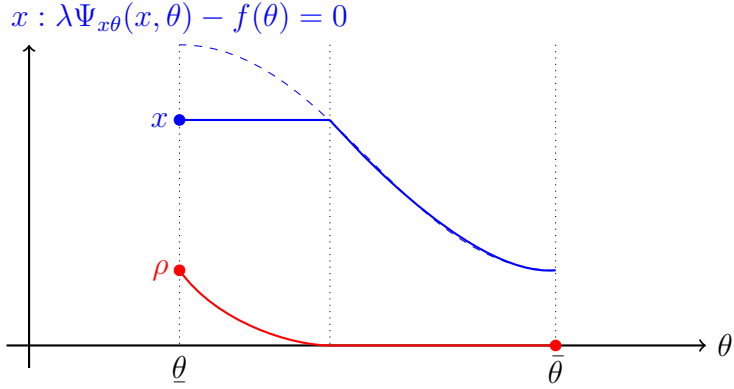
\begin{figure}
    \centering
    \begin{tikzpicture}

    \draw[->, thick] (-0.3,0) -- (9,0) node[anchor=west]{$\theta$};
    \draw[->, thick] (0,-0.3) -- (0,4);
    \draw[dotted] (2,4) -- (2,0) node[anchor=north]{$\ubar{\theta}$};
    \draw[dotted] (7,4) -- (7,0) node[anchor=north]{$\bar\theta$};

    \draw[dashed, blue] (7,1) .. controls (5,1) and (4,4) .. (2,4) node[anchor=south]{$x: \lambda \Psi_{x\theta}(x, \theta) - f(\theta) = 0$};

    \filldraw[red] (2, 1) circle (2pt) node[anchor=east]{$\rho$};
    \filldraw[blue] (2, 3) circle (2pt) node[anchor=east]{$x$};

    \draw[blue, thick] (2,3) -- (4,3);
    \draw[red, thick] (2,1) .. controls (2.5, 0.3) and (3.5,0) .. (4,0);
    \draw[dotted] (4,4) -- (4,0);

    \draw[blue, thick] (4,3) .. controls (4,3) and (5.8,0.9) .. (7,1);
    \draw[red, thick] (4,0) -- (7,0);
    \filldraw[red] (7, 0) circle (2pt);
    
\end{tikzpicture}
    \caption{Illustration of the conditions in Theorem \ref{thm:general_mech}. The dotted blue line satisfies the implicit equation. When $\rho$ is zero, $x$ must coincide with the dotted blue line.}
    \label{fig:general_mech_illustration}
\end{figure}

A remark on the convexity assumption on $\Psi_\theta$ is in order. The result in Theorem \ref{thm:general_mech} relies on the Pontryagin maximum principle providing a unique solution; this occurs given that the maximized Hamiltonian is concave. In the model, the Arrow sufficiency condition is that for all $\theta$,
\[ xf(\theta) - \lambda \Psi_\theta(x,\theta) \]
is concave in $x$. Assuming $\Psi_\theta$ is convex in $x$ is equivalent to this sufficiency condition. \footnote{There are other sufficiency conditions for Pontryagin's maximum principle with weaker assumptions on the Hamiltonian. Relaxing assumptions on $\Psi$ would require another sufficiency argument for the uniqueness of a solution characterized by the maximum principle.}

In the next sections, we discuss qualitative properties of the optimal mechanism and economic insights derived from the solution characterization in Theorem \ref{thm:general_mech}.

\subsection{Pooling}
We first discuss when the principal awards the maximum $T$ to the agent (i.e., for what types $t(\theta) = T$). From the budget binding constraint, we know that the most efficient type $\ubar{\theta}$ always receives the full transfer. In this section, we provide three insights. First, we show that a positive measure of types always receive the entire transfer; we say the optimal mechanism ``pools'' the most efficient types. Second, we characterize the threshold type; agents more efficient than the threshold are pooled and receive the full transfer $T$, and the principal optimally withholds the resource from agents less efficient than the threshold. Third, we characterize when the optimal mechanism pools all types; in this case, the optimal mechanism simply pays the entire $T$ to the agent whenever the agent meets a base action level. 

\begin{proposition}[Pooling At The Top]\label{prop:pooling_at_top}
    There exists a $\hat{\theta} > \ubar{\theta}$ such that the optimal mechanism $x^*$ is constant on $[\ubar{\theta}, \hat{\theta}]$. The transfer $t^*$ is equal to $T$ for these types.
\end{proposition}
\begin{proof}
Consider the conditions in Theorem \ref{thm:general_mech}. Since $\rho(\ubar{\theta}) > 0$, complementary slackness \eqref{cond:comp_slackness} implies that $x^*$ must be constant until $\rho$ reaches zero. Corollary \ref{corr:subsidy_schedule} implies that the optimal $t^*$ is also constant on this set of types, and the budget binding condition \eqref{cond:budget_bind} implies that the transfer is $T$.
\end{proof}

To intuitively understand why this pooling region always occurs, consider the loss in principal utility that occurs in the optimal mechanism relative to perfect information. With perfect information, the principal always gives away the entirety of the budget $T$ and obtains the maximum action $x$ from the agent such that $\Psi(x,\theta) = T$. The optimal mechanism, relative to perfect information, differs in two ways. Either the principal cannot obtain the highest action because the principal deliberately withholds the resource, or the principal must concede information rent to the agent. Hence, the trade-off between the relative loss from these two effects dictates the size of the pooling region of the most efficient types.

\begin{proposition}[Pooling Threshold]\label{prop:cutoff_type}
    Suppose the optimal $x^*$ is not constant on the entire interval of types. Let $\ubar{x} \equiv x^*(\ubar{\theta})$. Then the largest type $\hat{\theta}$ that receives the full transfer satisfies the following:
    \begin{align}\label{eqn:threshold}
    \frac{\Psi_{x\theta}(\ubar{x}, \hat{\theta})}{\Psi_x(\ubar{x}, \hat{\theta})} &=  \frac{f(\hat{\theta})}{F(\hat{\theta})}
    \end{align}
\end{proposition}
\begin{proof}
    By optimality, there exists some Lagrange multiplier and costate function $\lambda, \rho$ that satisfy the conditions of Theorem \ref{thm:general_mech}. At the cutoff, since $x^*$ is not constant on the entire interval, it must be that $x^*$ stops being constant at $\hat{\theta}$. Complementary slackness \eqref{cond:comp_slackness} implies that $\rho(\hat{\theta}) = 0$. Integrating out the costate evolution \eqref{cond:costate_evol} gives 
    \[ \lambda \Psi_x(\ubar{x},\hat{\theta}) = F(\hat{\theta}) \]
    Since we assumed $f$ was differentiable, $x^*$ is continuous at $\hat{\theta}$ but the right derivative must be negative; hence, $\dot{\rho}(\hat{\theta}) = 0$ since complementary slackness implies $\rho$ must be constant at zero for some interval above $\hat{\theta}$. Using the costate evolution \eqref{cond:costate_evol} and taking the upper limit $\theta \to \hat{\theta}$, we get 
    \[  \lambda \Psi_{x\theta}(\ubar{x},\hat{\theta}) = f(\hat{\theta}) \]
    Substituting out $\lambda$ gives \eqref{eqn:threshold}.
\end{proof}

To interpret this condition, 
rewrite \eqref{eqn:threshold} as \[ \Psi_x(\ubar{x}, \hat{\theta}) = \Psi_{x\theta}(\ubar{x}, \hat{\theta}) \frac{F(\hat{\theta})}{f(\hat{\theta})} \]
Under this rewriting, the condition is equivalent to the marginal benefit from paying the threshold type to be equal to the marginal information rent cost demanded.

It is possible that the principal always awards the entire $T$ to the agent, regardless of the type. The following result characterizes when this happens. Let $\bar{x}$ be such that $T = \Psi(\bar{x}, \bar{\theta}) $ (that is, the highest action that least efficient type $\bar{\theta}$ would be willing to take if given the entire budget $T$). Then
\begin{proposition}[Complete Pooling]\label{prop:all_pooling}
    The optimal mechanism pools all types (that is, the optimal $x$ is constant) if and only if for all $\theta$, 
    \[ \Psi_x(\bar{x}, \theta) \ge F(\theta)\Psi_x(\bar{x}, \bar{\theta}). \]
    Further, if the optimal mechanism pools all types, then $x(\theta) = \bar{x}$ and $t(\theta) = T$ for all $\theta$.
\end{proposition}
\begin{proof}
First, suppose that $x^*(\theta) = x^*$ is constant. Let $\lambda, \rho$ be the Lagrange multiplier and costate variable satisfying Theorem \ref{thm:general_mech}. By the budget binding condition \eqref{cond:budget_bind}, 
\begin{align*} 
T &= \Psi(x^*, \ubar{\theta}) + \int_{\ubar{\theta}}^{\bar{\theta}} \Psi_\theta(x^*, s) \ ds \\
&=  \Psi(x^*, \ubar{\theta}) + \Psi(x^*, \bar{\theta}) - \Psi(x^*, \ubar{\theta})  \\
&= \Psi(x^*, \bar{\theta})
\end{align*}
Hence, $\bar{x} = x^*$. By the upper bound condition \eqref{cond:upper_bound}, integrating the costate evolution equation \eqref{cond:costate_evol}, we must have 
\[\rho(\bar{\theta}) = \lambda \Psi_{x}(\bar{x}, \bar{\theta}) - 1 = 0 \]
So $\lambda = 1/\Psi_x(\bar{x}, \bar{\theta})$. Then since $\rho$ satisfies the costate evolution condition \eqref{cond:costate_evol},
\[ \rho(\theta) = \lambda \Psi_x(\bar{x}, \theta) - F(\theta) \ge 0 \]
and so the result follows.

For the converse, suppose that $\bar{\lambda} \Psi_x(\bar{x}, \theta) \ge F(\theta)$ for all $\theta$. Then construct $\rho(\theta) = \lambda \Psi_x(\bar{x}, \theta) - F(\theta)$, and define $x^*(\theta) = \bar{x}$ and $\lambda = 1/\Psi_x(\bar{x}, \bar{\theta})$. It is easy to check that these satisfy the condition of Theorem \ref{thm:general_mech}. Hence the optimal mechanism is $x(\theta) = \bar{x}$ and $t(\theta) = T$ for all $\theta$.
\end{proof}

\paragraph{Example Revisited (Emissions Reduction)} Revisit the emissions reduction example from before. As Corollary \ref{corr:subsidy_schedule} implies, the optimal mechanism for the regulator is to offer a subsidy schedule, that pays some amount to each firm depending on their emissions reduction. Proposition \ref{prop:pooling_at_top} implies that the regulator pays the maximum amount $T$ for a sufficiently large reduction, which is taken by a \textit{positive measure} of firms. Proposition \ref{prop:all_pooling} then provides conditions for which the optimal regulator's policy has a ``bang-bang'' feature; that is, the regulator never gives any partial transfer, but only pays $T$ or zero. This latter case is equivalent to providing an unconditional transfer and mandating a fixed level of emissions reduction.

\subsection{Shadow Costs}
Having established that the mechanism always admits a pooling region up to some $\hat{\theta}$, we provide some economic intuition for the optimal $x$ given $\hat{\theta}$. That is, throughout this section, suppose the optimal mechanism has a cutoff type $\hat{\theta} < \bar{\theta}$ determined by Proposition \ref{prop:cutoff_type}. We characterize the schedule above and below the cutoff, using the shadow value interpretation of the Lagrange multiplier $\lambda$.

First, define $\bar{x} \equiv x(\theta)$ for $\theta \in [\ubar{\theta}, \hat{\theta}]$. The following result intuitively shows that the optimal $\bar{x}$ solves a maximization problem for the principal, where the marginal benefit is given by the probability that the agent meets the cutoff, and the marginal cost is $\lambda \Psi(\cdot,\hat{\theta})$, the shadow cost of the transfer that must be paid to the threshold type.
\begin{corollary}\label{corr:highest_action}
    Given the optimal shadow value of money $\lambda$, the highest action $\bar{x}$ solves the following maximization:
    \[ \bar{x} = \arg \max_x \left\{ x F(\hat{\theta}) - \lambda \Psi(x,\hat{\theta}) \right \}. \]
\end{corollary}
\begin{proof}
    Integrating out the costate evolution from Theorem \ref{thm:general_mech}, we find that 
    \[ \lambda \Psi_x(\bar{x},\hat{\theta}) - F(\hat{\theta}) = 0 \]
    which is exactly the FOC for the maximization. The FOC exactly characterizes the solution since $\Psi$ is convex in $x$. 
\end{proof}

Corollary \ref{corr:highest_action} characterizes the tradeoff for the highest action $\bar{x}$; intuitively, by increasing the highest action, the principal gets more production from $F(\hat{\theta})$ types of agents, but pays the shadow cost $\lambda \Psi(\cdot, \hat{\theta})$. As mentioned before, the value of the Lagrange multiplier $\lambda$ can be interpreted as the shadow value of the resource to the principal.

Now, we show that the types from $(\hat{\theta}, \bar{\theta}]$ intuitively face a \textit{different} shadow cost to the principal. This insight helps provide us an intuition for what the optimal mechanism looks like for types $\theta > \hat{\theta}$. 

\begin{proposition}\label{prop:virtual_budget_cost}
    Suppose $\Psi_{x\theta}(x, \theta) / f(\theta)$ is weakly increasing in $\theta$ and $\hat{\theta} < \bar{\theta}$. Then on $[\hat{\theta}, \bar{\theta}]$, the optimal $x$ satisfies the point-wise maximization:
    \begin{equation}\label{eqn:shadow_optimization}
        x(\theta) = \max_{x} \left\{ x f(\theta) - \lambda \Psi_{\theta}(x, \theta) \right\} 
    \end{equation} 
\end{proposition}
\begin{proof}
    Recall that by definition, since pooling ends at $\hat{\theta}$, $\rho(\hat{\theta}) = 0$. By the Bellman principle of optimality, the restriction of $x$ to the domain $[\hat{\theta}, \bar{\theta}]$, must be optimal for the following control subproblem:
    \begin{align*}
& \max_{x} && \int_{ \hat{\theta}}^{\bar{\theta}} \left( x(\theta) f(\theta) - \lambda \Psi_\theta(x(\theta), \theta) \right)  \ d\theta  \\
& \textnormal{subject to} && \dot{x} = u \le 0 &&&&\textnormal{(Monotonicity)} 
\end{align*}
Note that the solution to \eqref{eqn:shadow_optimization} point-wise maximizes the objective of the subproblem. To confirm that this is optimal, it suffices to show that $x(\theta)$ satisfies the monotonicity constraint. Note that since $\Psi_\theta$ is convex, the first-order condition for $x$ gives 
\[ 1 = \lambda \frac{\Psi_{x\theta}(x, \theta)}{f(\theta)} \]
Note that the RHS is an increasing function of $x$. Since the RHS is assumed to be increasing in $\theta$ and increasing in $x$, it follows that $x(\theta)$ must be decreasing in $\theta$.
\end{proof}

To interpret the maximization problem in \eqref{eqn:shadow_optimization}, note that by increasing $x(\theta)$, the principal gains a marginal benefit of $f(\theta)$ but must pay $\Psi_\theta(x(\theta),\theta)$ more \textit{information rent} to the best types (or equivalently, must procure $\Psi_\theta(x(\theta),\theta)$ of the resource). Note that the shadow resource cost is $\lambda$, so the cost is exactly the shadow information rent cost that must be paid. 

As a remark, note that the main obstacle to point-wise maximization is the monotonicity constraint; hence, one can obtain a more general variant of Proposition \ref{prop:virtual_budget_cost} by ironing the shadow cost function $\Phi_\theta(x,\theta)/f(\theta)$ (and in fact, this is what the costate variable does in Theorem \ref{thm:general_mech}). 

In sum, the optimal mechanism can be interpreted as the principal splitting the agent types into two populations, and maximizing $x$ against different shadow cost functions. For the efficient types, the principal maximizes against the shadow production cost, $\lambda \Psi$; for the less efficient types, the principal maximizes against the shadow information rent cost, $\lambda \Psi_\theta$. 

\subsection{Separable Cost Function}
In this section, we show that for multiplicatively separable $\Psi$ the mechanism admits a closed form solution, where the shape of the solution, including the pooling region from Proposition \ref{prop:pooling_at_top}, is determined by the belief. In particular, the multiplicatively separable case is remarkable in that the size of the pooling region is \textit{independent} of $T$, which we show in Proposition \ref{prop:T_invariance}.

Firstly, we say $\Psi$ is separable in $x$ and $\theta$ if $\Psi(x,\theta) = B(\theta) \Gamma(x)$, where $\Gamma$ is convex by assumption. Without loss, since $\Psi$ is assumed increasing in $\theta$, we can renormalize the type space and redefine $\Theta$ so that $\Psi(x,\theta) = \theta \Gamma(x)$ where $\Gamma$ is a strictly convex function. In this case, we can actually provide a closed-form characterization of the shape of the schedule. Let $\tilde{f}$ denote the left derivative of $\textnormal{cav }F$, the concavification of the cumulative distribution function $F$ on $[0, \infty)$, and denote by $(\Gamma')^{-1}$ the inverse function associated with $\Gamma'$ (which is increasing by assumption).
\begin{proposition}\label{prop:separable}
    Suppose that $\Psi(x,\theta) = \theta \Gamma(x)$, where $\Gamma$ is convex. Then the optimal mechanism $(x^*, t^*)$ induces a production schedule
    \[ x^*(\theta)= (\Gamma')^{-1}\left(\frac{\tilde{f}(\theta)}{\lambda}\right) \]
    where $\lambda$ is the nonnegative Lagrange multiplier chosen so that
    \[ \ubar{\theta}\Gamma(x^*(\ubar{\theta})) + \int_{\ubar{\theta}}^{\bar{\theta}} \Gamma(x^*(\theta)) \ d\theta = T \]
    The transfer schedule is still given by Lemma \ref{lem:general_feasible}.
\end{proposition}

Intuitively, this holds because the linearity of $\Psi$ in $\theta$ allows the separation of the monotonicity constraint into a standard ironing problem. Thus, for this specific case, the optimal schedule depends on the density of the concavified cumulative distribution function of the belief, $\tilde{f}$. 
Since Proposition \ref{prop:separable} shows that $x^*$ is an increasing transformation of $\tilde{f}$, the content of Proposition \ref{prop:separable} is that for separable cases, the shape of the optimal mechanism is entirely determined by the \textit{belief} of the principal and not by the agent's cost function. Indeed, the pooling behavior of the optimal mechanism is completely identical for $\Psi^A(x,\theta) = \theta \Gamma^A(x)$ and $\Psi^B(x,\theta) = \theta \Gamma^B(x)$ for any two convex $\Gamma^A, \Gamma^B$.

We illustrate how this insight applies to a discrete, two-type example. Suppose $\theta$ is either $1$ or $2$, and the cost function takes the form $\Psi(x,\theta) = \theta x^2$. The principal believes that $\theta = 1$ with probability $\mu$. Figure \ref{fig:separable_illustration} illustrates $\tilde{f}$; for the two-type case, whether $\mu$ is larger than or less than $1/2$ determines whether $\tilde{f}(1) = \tilde{f}(2)$ or not. By Proposition \ref{prop:separable}, $x^*$ is an increasing transformation of $\tilde{f}$, so the shape of $\tilde{f}$ also determines the shape of $x^*$. The optimal mechanism separates types 1 and 2 if and only if $\mu > 1/2$. \footnote{In our modeling assumptions, we took $F$ as continuous and differentiable. We can extend this to scenarios where $F$ is discontinuous (equivalently, $f$ admits point masses) by taking the limits of a sequence of continuous $F_n$ converging to $F$ and observing that that for the separable case, the optimizer function is continuous in $\tilde{f}$, which is continuous in $f$ in the sup norm.}

\begin{figure}
\begin{subfigure}[b]{0.48\textwidth}
\centering
\begin{tikzpicture}
\draw[thick, <->] (-3,0) -- (3,0) node[anchor=south]{$\theta$};
\draw[dotted] (-2, 3) -- (-2, 0) node[anchor=north]{$0$};
\draw[dotted] (0, 3) -- (0, 0) node[anchor=north]{$1$};
\draw[dotted] (2, 3) -- (2, 0) node[anchor=north]{$2$};
\filldraw[blue] (-2, 0) circle (2pt);
\draw[blue, ultra thick] (-2,0) -- (0,0);
\filldraw[blue] (0, 1) circle (2pt);
\draw[blue] (0, 0) circle (2pt);
\draw[blue, ultra thick] (0,1) -- (2,1);
\filldraw[blue] (2, 3) circle (2pt);
\draw[blue] (2,1) circle (2pt);
\draw[blue, ultra thick, ->] (2,3) -- (3,3);
\draw[red, ultra thick, dashed] (-2,0) -- (2,3);
\end{tikzpicture}
\caption{ $\mu \le 1/2$  } 
\label{fig:pooling_twotype}
\end{subfigure}
\begin{subfigure}[b]{0.48\textwidth}
\centering
\begin{tikzpicture}
\draw[thick, <->] (-3,0) -- (3,0) node[anchor=south]{$\theta$};
\draw[dotted] (-2, 3) -- (-2, 0) node[anchor=north]{$0$};
\draw[dotted] (0, 3) -- (0, 0) node[anchor=north]{$1$};
\draw[dotted] (2, 3) -- (2, 0) node[anchor=north]{$2$};
\filldraw[blue] (-2, 0) circle (2pt);
\draw[blue, ultra thick] (-2,0) -- (0,0);
\filldraw[blue] (0, 2.3) circle (2pt);
\draw[blue] (0, 0) circle (2pt);
\draw[blue, ultra thick] (0,2.3) -- (2,2.3);
\filldraw[blue] (2, 3) circle (2pt);
\draw[blue] (2,2.3) circle (2pt);
\draw[blue, ultra thick, ->] (2,3) -- (3,3);
\draw[red, ultra thick, dashed] (-2,0) -- (0,2.3);
\draw[red, ultra thick, dashed] (0,2.3) -- (2,3);
\end{tikzpicture}
\caption{ $\mu > 1/2$  } 
\label{fig:separating_twotype}
\end{subfigure}
\caption{The distribution function corresponding to the belief of the principal is plotted in solid lines, and the dashed line denotes the concavification. In the first case, the concavified distribution function is constant, and so the optimal mechanism pools the types. In the second case, the left-derivative of the concavified distribution function takes different values at $1$ and $2$, so the optimal mechanism separates types.}\label{fig:separable_illustration}
\end{figure}
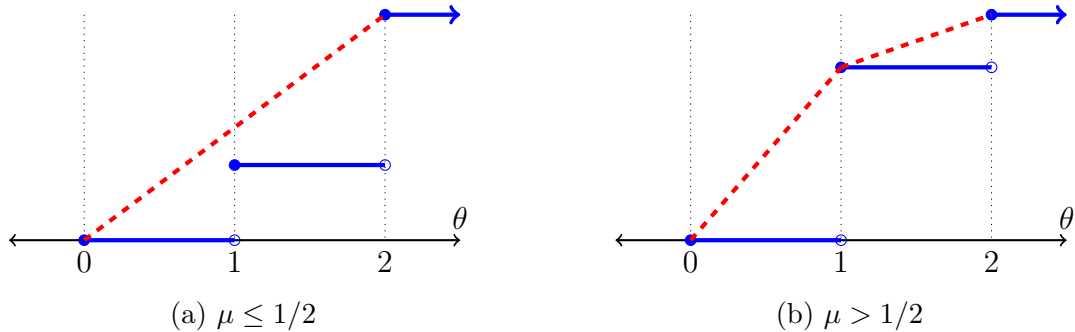

Note that for the separable case, Propositions \ref{prop:cutoff_type} and \ref{prop:all_pooling} simplify considerably:

\begin{corollary}\label{corr:sep_cutoff_type}
    Suppose $\Psi(x,\theta) = \theta \Gamma(x)$, and suppose the optimal $x^*$ is not constant on the entire interval of types. Let $\ubar{x} \equiv x^*(\ubar{\theta})$. Then the largest type $\hat{\theta}$ that receives the full transfer satisfies the following:
    \begin{align}\label{eqn:sep_threshold}
    \frac{1}{\hat{\theta}} &=  \frac{f(\hat{\theta})}{F(\hat{\theta})}
    \end{align}
\end{corollary}
\begin{corollary}\label{corr:sep_all_pooling}
    Suppose $\Psi(x,\theta) = \theta \Gamma(x)$. Then the optimal mechanism pools all types if and only if for all $\theta$, 
    \[ \theta / \bar{\theta} \ge  F(\theta) \]
\end{corollary}

In some sense, the upper bound on the schedule implied by Proposition \ref{prop:pooling_at_top} follows from the existence of an upper bound on transfers. An initial intuition might be that the pooling region occurs because the constraint is binding, and that the binding region shrinks if $T$ is large. However, this is \textit{not} the case. The pooling arises because of the \textit{relative} informational cost of separating low and high $\theta$ to \textit{relative} production ability. In fact, for the separable case, we have the following result:

\begin{proposition}\label{prop:T_invariance}
    Suppose $\Psi(x,\theta) = \theta \Gamma(x)$. Let $\hat{\theta}$ be the largest type that receives the full transfer. Then $\hat{\theta}$ is invariant to $T$ and $\Gamma$.
\end{proposition}
\begin{proof}
    By Proposition \ref{prop:separable}, consider $x_T$ and $x_{T'}$, the optimal schedules for budgets $T$ and $T'$. Since $\tilde{f}$ is unchanged, and $x^*$ is an increasing transformation of $\tilde{f}$, the pooling region for $x_T$ and $x_{T'}$ are the same, and hence $\hat{\theta}$ is invariant to altering the size of the budget.
\end{proof}

That is, the size of the pooling region does not change at all as $T$ changes. This implies that the \textit{shape} of the contract is determined primarily by the belief of the principal, and the \textit{scale} of the contract is determined by the size of the budget. Hence, the existence of the pooling region at the top stems from the principal's relative tradeoff between giving up information rents and withholding a non-valuable resource. 

\subsection{Linear Subsidy Schedule}

Suppose the principal only offered some linear subsidy, paying for $x$ at a constant rate $r$. By setting such a subsidy schedule, very efficient types exhaust the budget and take $x = T/r$ and less efficient types choose some $x$ that equalizes their marginal cost with the rate: $\Psi_x(x, \theta) = r$. We will show that the mechanism design approach can always do better; that is, some feasible schedule Pareto dominates the schedule resulting from a linear subsidy.

\begin{proposition} \label{prop:linear_schedule_bad}
Let $x_r$ be the outcome schedule from a linear subsidy at rate $r$. There exists a feasible schedule $x^*$ in the mechanism design problem such that for all $\theta$, $x^*(\theta) \ge x_r(\theta)$ with inequality holding strictly on a positive measure of $\theta$.
\end{proposition}

From Proposition \ref{prop:linear_schedule_bad}, we thus see that for any outcome implementable by a linear subsidy schedule, there is a feasible schedule implementable in the mechanism design problem that yields more output. So the linear subsidy schedule is dominated by the optimal outcomes that are implementable in the mechanism design problem.

\section{Linear Value}\label{sec:linear}
In the baseline model, we supposed that the principal's objective was independent of $t$, which we interpret as the principal having no value for the motivational resource. In certain applications, it may be reasonable to suppose the principal does have some value for the resource (or alternatively, obtaining the resource incurs a cost to the principal per unit of resource). 

To motivate the model, we revisit our primary application, the emissions reduction example from before.

\paragraph{Example (Emissions Reduction)} Let us return to the policymaker trying to induce a large number of firms to take reduce emissions. Now, in addition to providing a subsidies for emissions reductions, the policymaker could also use some of the funds to purchase emissions offsetting credits in a competitive market, where the market price for offsetting a unit of emissions is some constant $k$. The policymaker would like to maximize the net emissions reduction across both channels, via offsetting or reducing emissions.

Formally, consider a variation of the design problem \eqref{prblm:general}:
\begin{align}
& \max_{x, t} && \int_{\theta \in \Theta} \left[  x(\theta) - k t(\theta) \right]f(\theta) \ d\theta \label{prblm:linear} \\
& \textnormal{subject to} && t(\theta) - \Psi(x(\theta),\theta) \ge  t(\theta') - \Psi(x(\theta'),\theta) && \forall \theta, \theta'\in \Theta && \textnormal{(IC)} \notag \\
&&& t(\theta) - \Psi(x(\theta),\theta) \ge 0 && \forall \theta \in \Theta  && \textnormal{(IR)} \notag \\ 
&&& t(\theta) \le T && \forall \theta \in \Theta &&\textnormal{(B)} \notag 
\end{align}
where $k \ge 0$ is some constant. That is, the important distinction is in the objective, where there is an additional $-k t(\theta)$ term that appears. This term could correspond to the principal obtaining some value from the unused $T-t$ measure of resource, or the principal incurring some explicit cost to obtain the resource, with this problem degenerating to the original problem when $k=0$. We can characterize the optimal solution as follows:
\begin{theorem}\label{thm:linear_cost}
    An optimal mechanism $(x, t)$ for design problem \eqref{prblm:linear} exists, and it satisfies Lemma \ref{lem:general_feasible}. There exists a nonnegative Lagrange multiplier $\lambda$, a nonnegative, absolutely continuous costate function $\rho$, and a nonnegative, absolutely continuous costate $w$, such that $(x,t,\lambda,\rho,w)$ uniquely solve complementary slackness \eqref{cond:comp_slackness}, border conditions \eqref{cond:upper_bound} and \eqref{cond:lower_bound}, a modified costate evolution given by 
     \begin{equation} \label{cond:costate_linear}
         \dot{\rho}(\theta) =  \lambda \Psi_{x\theta}(x^*(\theta), \theta) - f(\theta)\left( 1 - k\left( \Psi_{x}(x^*(\theta),\theta) + \Psi_{x\theta}(x^*(\theta),\theta)\frac{F(\theta)}{f(\theta)} \right) \right) - w(\theta)
     \end{equation}   
    that holds wherever $\rho$ is differentiable, and an additional complementary slackness constraint
    \begin{equation}\label{cond:comp_slackness_x}
        w(\theta) > 0 \implies x(\theta) = 0. 
    \end{equation}
    Further, the budget binding condition \eqref{cond:budget_bind} must hold if $\lambda > 0$.
\end{theorem}

A couple of comments on the modifications are in order. Note that the altered costate evolution features two new terms. First, 
\begin{equation}
   \psi(x,\theta)=  \Psi_x(x,\theta) + \Psi_{x\theta}(x,\theta) \frac{F(\theta)}{f(\theta)}
\end{equation}
is the \textit{virtual} marginal cost of the agent, similarly to the Myersonian virtual value. The second, $w$, is a complementary slackness variable that arises because the lower bound on $x$ can bind and types of agents can be excluded in this case.

\subsection{Pooling}
Given the additional term in the objective of the problem, we show that some features of the baseline solution still hold. Firstly, we have the following analogues of Proposition \ref{prop:pooling_at_top} and \ref{prop:cutoff_type}:
\begin{proposition}\label{prop:linear}
    If $\lambda > 0$, there exists a $\hat{\theta} > \ubar{\theta}$ such that the mechanism is constant on $[\ubar{\theta},\hat{\theta}]$. The transfer $t^*$ is equal to $T$ on this interval. Further, if $\hat{\theta} < \bar{\theta}$, then $\hat{\theta}$ satisfies the following implicit equation:
    \begin{equation}\label{eqn:threshold_linear}
        \frac{\Psi_{x\theta}(\ubar{x}, \hat{\theta})}{\Psi_x(\ubar{x},\hat{\theta})\left( 1 - k \Psi_{x}(\ubar{x},\hat{\theta}) \right)} = \frac{f(\hat{\theta})}{F(\hat{\theta})}
    \end{equation}
    where $\ubar{x} = x(\ubar{\theta})$.
\end{proposition}
\begin{proof}
    If $\lambda > 0$, since the budget must bind, $x(\ubar{\theta}) > 0$, and so it follows that $\rho(\ubar{\theta}) > 0$, so complementary slackness implies that $x$ must be constant and positive. Integrating out the new costate evolution, we get that since $x(\ubar{\theta}) > 0$, $w(\theta) = 0$, and so 
    \[ \rho(\theta) = \lambda \Psi_x(\ubar{x},\theta) - F(\theta) (1 - k \Psi_x(\ubar{x},\theta))  \]
    At the threshold $\hat{\theta}$, $\rho(\hat{\theta}) = 0$, and by a similar argument as Proposition \ref{prop:cutoff_type}, $\dot{\rho}(\hat{\theta}) = 0$, so 
    \[ \lambda \Psi_{x\theta}(\ubar{x}, \hat{\theta}) = f(\hat{\theta})\left( 1 - k\left( \Psi_{x}(\ubar{x},\hat{\theta}) + \Psi_{x\theta}(\ubar{x},\hat{\theta})\frac{F(\hat{\theta})}{f(\hat{\theta})} \right) \right) \]
    \[ \lambda + F(\hat{\theta})k =  \frac{f(\hat{\theta})}{\Psi_{x\theta}(\ubar{x}, \hat{\theta})}\left( 1 - k \Psi_{x}(\ubar{x},\hat{\theta}) \right) \]
    Combining, we get \eqref{eqn:threshold_linear}.
\end{proof}
With a linear value, the pooling region depends on whether the budget constraint binds. For sufficiently high values of $k$ or high values of $T$, the optimal schedule may be such that the principal's optimal schedule in the absence of the budget constraint is actually feasible. Note that this could not happen in the baseline model, where the budget always binds; since the principal did not obtain any value from the resource in the baseline case, the baseline optimization problem without a budget constraint is unbounded (giving infinite resource to the agent). 

Also, with the principal having a value for the resource, it is possible that a positive measure of types are excluded from the mechanism. We provide a sufficient condition for such exclusion to happen.

\begin{proposition}
Suppose ${\Psi_{x\theta}(x, \theta)}/{f(\theta)}$ is increasing in $\theta$ and $k > 1/\psi(0, \bar{\theta})$. Then the optimal mechanism excludes a positive measure of types.
\end{proposition}
\begin{proof}Firstly, if $k > 1/\Psi_x(0, \ubar{\theta}) = 1/\psi(0, \ubar{\theta})$, it is not hard to check that the mechanism is degenerate: $x(\cdot) = t(\cdot) = 0$ (since intuitively, even the lowest cost type $\ubar{\theta}$ has a marginal cost of more than $1/k$). So suppose $k < 1/\psi(0, \ubar{\theta})$. By the costate evolution \eqref{cond:costate_linear}, on a sequence of $\theta \to \tilde{\theta}$ from the left, continuity of $x$ implies that $x$ is positive and decreasing. So $w$ must be 0 and $\rho = 0$, so  
     \[ 0 = \lambda \Psi_{x\theta}(0, \tilde{\theta}) - f(\tilde{\theta}) \left( 1 - k \left(\Psi_{x}(0, \tilde{\theta}) + \Psi_{x\theta}(0, \tilde{\theta})\frac{F(\tilde{\theta})}{f(\tilde{\theta})} \right) \right)\]
     \[    1 - k\psi(0, \tilde{\theta})   =  \lambda \frac{\Psi_{x\theta}(0, \tilde{\theta})}{f(\tilde{\theta})} \]
Note by our assumptions, the RHS is increasing, and the LHS is decreasing in $\tilde{\theta}$. Further, the LHS goes from positive to negative from $\ubar{\theta}$ to $\bar{\theta}$ and the RHS is positive, so the intermediate value theorem implies there exists a $\tilde{\theta} < \bar{\theta}$ such that the equality holds. Then the optimal mechanism excludes all types $\theta > \tilde{\theta}$. \end{proof}

Intuitively, this arises because the principal has a positive value for the resource; thus, intuitively, if some type of agent has an information rent-adjusted marginal cost of $x$ as more than $1/k$, the principal finds it best to exclude this agent.

\subsection{Comparison to Naive Benchmark}
Recall in the baseline model, the problem was unbounded without the budget constraint, since the principal did not value the resource. With a linear value for the resource $k$, the problem is no longer unbounded; thus, we can compare the optimal mechanism to a naive solution, where the principal designs the mechanism (i.e. solves \eqref{prblm:linear}) without factoring the budget constraint, and instead ``runs out'' of money if the mechanism is supposed to offer $t > T$. 

The solution to the naive problem corresponds to the problem where the Lagrange multiplier $\lambda$ is set to zero. As a result, from Theorem \ref{thm:linear_cost}, the condition pinning down $x$ when not in a flat region is missing a $\lambda \Psi_{x\theta}$ term, which implies that $x$ is larger when the budget constraint does not bind. In consequence, the budget binds at a higher type than in the optimal mechanism (i.e. the naive solution spends more of the resource in expectation), and the naive solution obtains a lower $x$ from the most efficient types of agents. Figure \ref{fig:naive_comparison} illustrates this with an example, where values are distributed uniformly on $[1,2]$, $k=0.3$, the cost function is $\Psi(x,\theta) = \theta x^2$, and the budget is $T=1$. 

\begin{figure}
     \centering
     \begin{subfigure}[b]{0.48\textwidth}
         \centering
         \includegraphics[width=\textwidth]{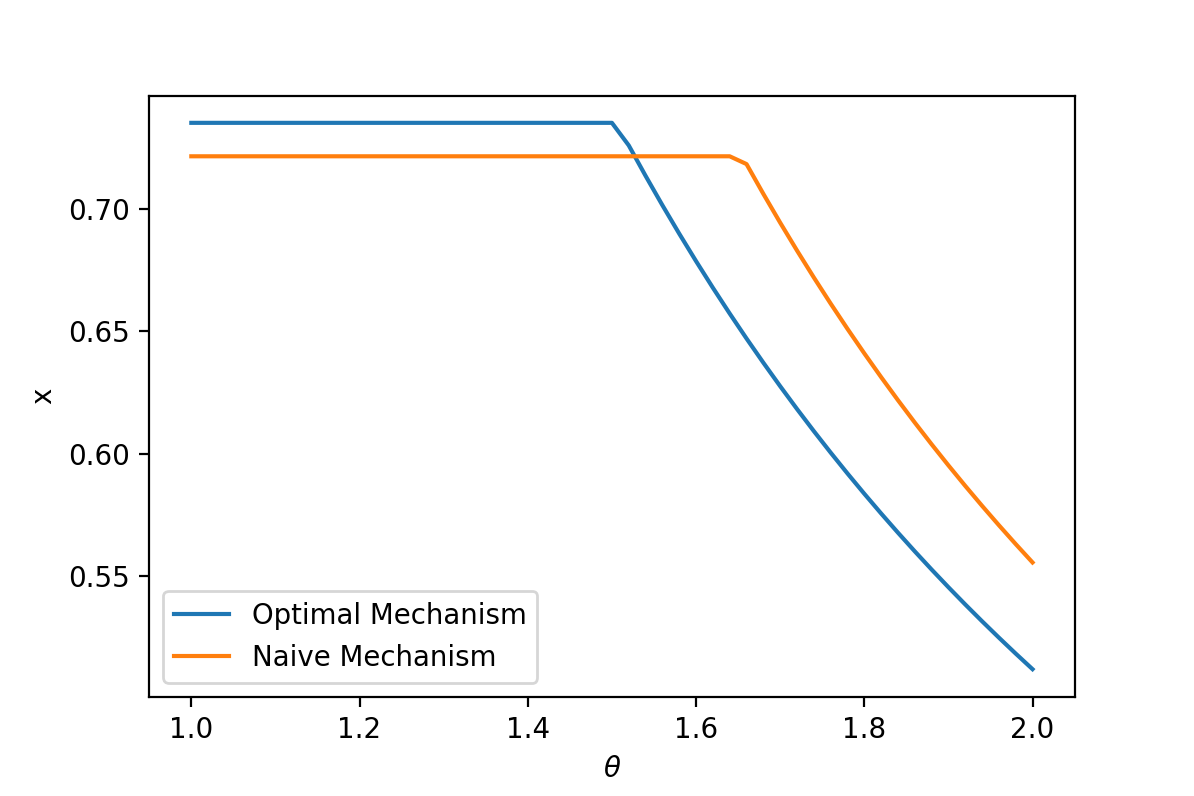}
         \caption{Action schedule $x$}
         \label{fig:naive_x}
     \end{subfigure}
     \hfill
     \begin{subfigure}[b]{0.48\textwidth}
         \centering
         \includegraphics[width=\textwidth]{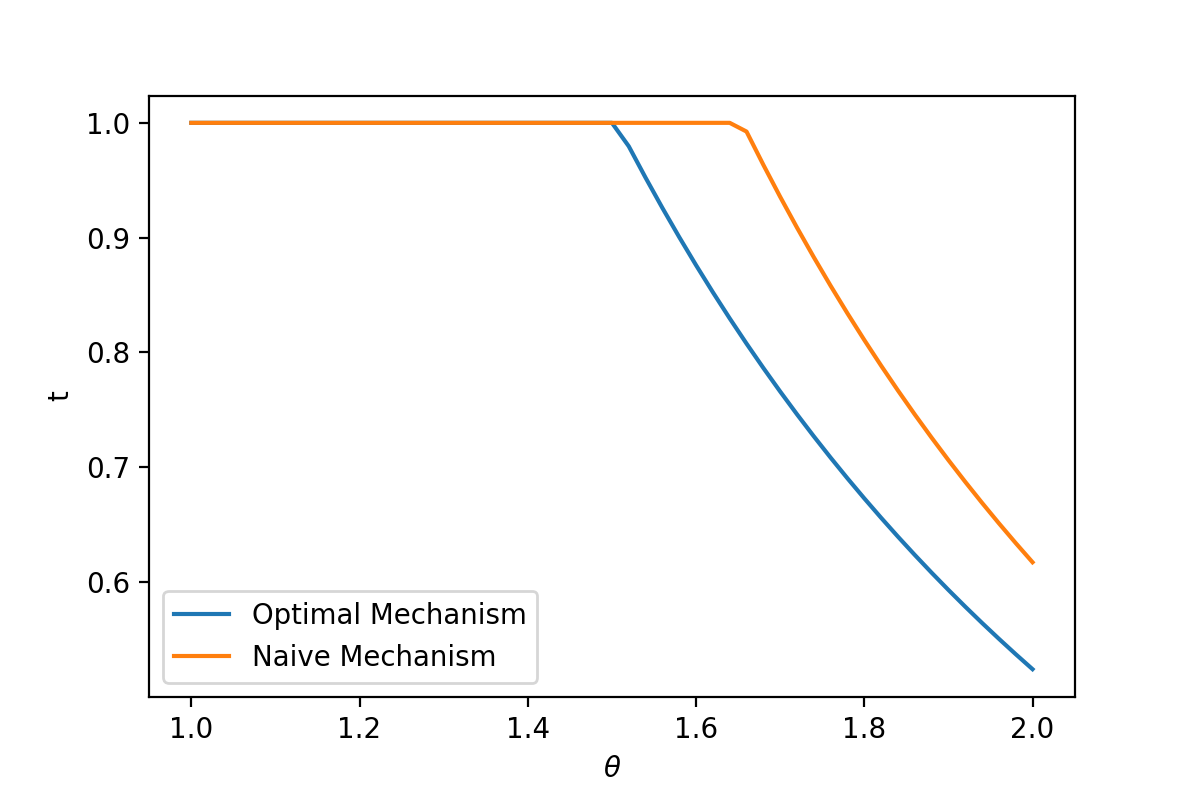}
         \caption{Transfer schedule $t$}
         \label{fig:naive_t}
     \end{subfigure}
        \caption{Comparison of naive and optimal mechanisms. Values are distributed uniformly on $[1,2]$, $k=0.3$, the cost function is $\Psi(x,\theta) = \theta x^2$, and the budget is $T=1$. The naive mechanism uses more of the resource in expectation, and obtains less $x$ from more efficient types than the optimal mechanism. }
        \label{fig:naive_comparison}
\end{figure}

\subsection{Comparative Statics}
We characterize how the optimal mechanism changes in terms of $k$ in the following result. Define $\bar{x} = x(\ubar{\theta})$; equivalently, $\bar{x}$ is the $x$ needed for the agent to get $T$ in the mechanism. 
\begin{proposition}\label{prop:linear_comparative}
As $k$ increases, $\bar{x}$ increases and $\lambda$ decreases: $d\bar{x}/dk > 0$ and $d\lambda / dk < 0$. 
\end{proposition}

We leave the proof to the appendix; the key idea is that we identify three equations from Theorem \ref{thm:linear_cost} that pin down $(\bar{x},\hat{\theta},\lambda)$, and apply the implicit function theorem to derive the comparative statics. To illustrate the implications of the proposition, Figure \ref{fig:linear_cost} plots the optimal mechanism for $\Psi(x,\theta) = \theta x^2$, $T=2$, and $\theta$ uniformly distributed on $[0,1]$. As $k$ increases, the $x$ required to receive the full transfer $T$ in the optimal mechanism decreases, and the mechanism excludes more types. However, since $\lambda$ increases in $k$, it also is true that $x(\bar{\theta})$ decreases; that is, the principal becomes more lenient on the less efficient types, reducing $x$ for higher types. 
 
Revisiting the emissions reduction example, the comparative statics in $k$ can be interpreted as changes in the optimal subsidy schedule as the emissions offset market becomes more or less effective. The comparative statics highlight the role of the threshold $\hat{\theta}$: as the outside option increases, the policymaker becomes more demanding (asks for more aggressive emissions reductions $x$) from firms that are more efficient than the threshold, but becomes more lenient (asks for less emissions reductions $x$) from types that are less efficient than the threshold. 

\begin{figure}
    \centering
    \includegraphics[width=0.5\textwidth]{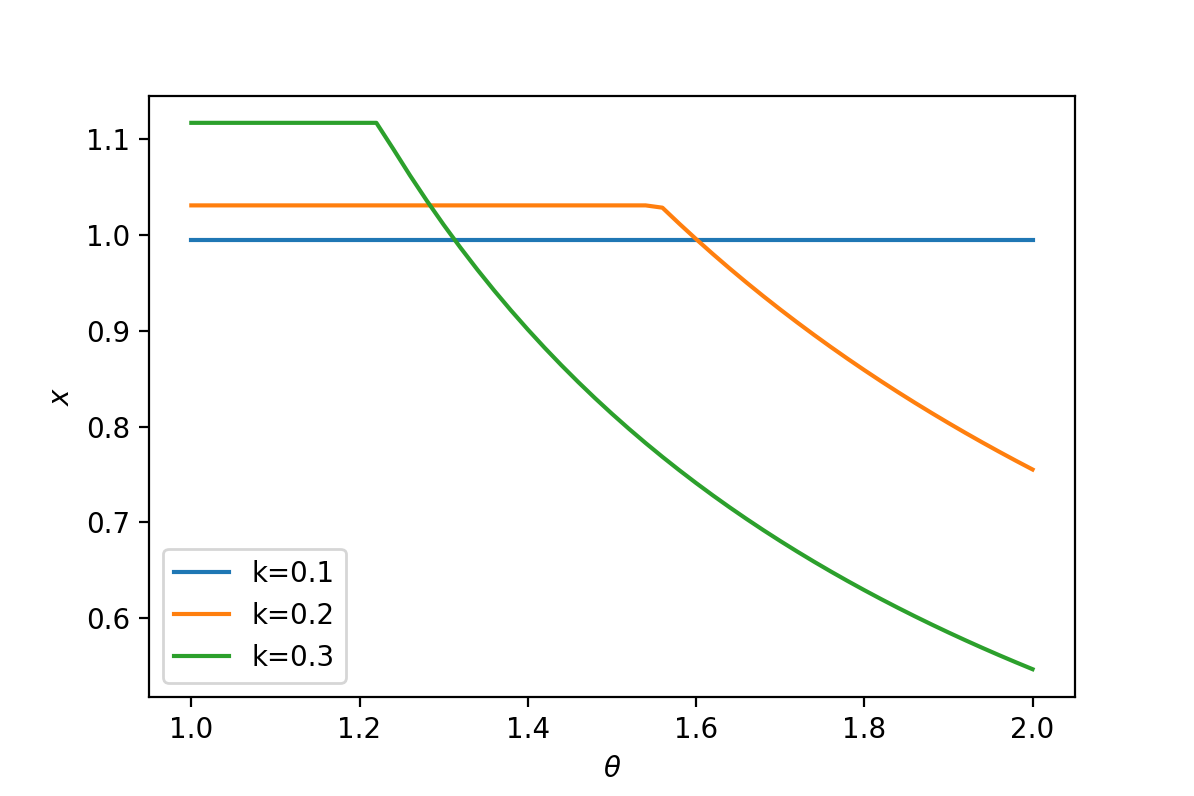}
    \caption{Optimal mechanisms for $\Psi(x,\theta) = \theta x^2$, $T=2$, $\theta$ uniformly distributed on $[1,2]$, as $k$ varies.}
    \label{fig:linear_cost}
\end{figure}

\section{Variants}
\label{sec:extensions}

In this section, we discuss some other extensions of our model. We first consider a generalization with multiple agents, and derive the control problem characterizing the optimal solution. 
We also characterize the problem under an ex-ante budget constraint, and contrast the optimal mechanism to the solution characterized by Theorem \ref{thm:general_mech}. 

\subsection{Multiple Agents}
Suppose the principal could contract with one of $N$ agents, each of whose private type is drawn independently from $F$. That is, consider the problem:
\begin{align}
& \max_{x, t} && \int_{\theta \in \Theta} 
 x(\theta) (1 - F(\theta))^{N-1} f(\theta) \ d\theta \label{prblm:multi} \\
& \textnormal{subject to} && t(\theta) - \Psi(x(\theta),\theta) \ge  t(\theta') - \Psi(x(\theta'),\theta) && \forall \theta, \theta'\in \Theta && \textnormal{(IC)} \notag \\
&&& t(\theta) - \Psi(x(\theta),\theta) \ge 0 && \forall \theta \in \Theta  && \textnormal{(IR)} \notag  \\
&&& t(\theta) \le T && \forall \theta \in \Theta && \textnormal{(B)} \notag 
\end{align}
This problem is a restricted version of a fully-general multi-agent problem; specifically, we restrict that the principal can only contract with the best (least cost type) of the agents, and must design the contract independently of the other agents' reports. Further, the incentive constraint (IC) requires the agent reporting the best-type to not want to deviate, conditional on being the best type. Note that under this specification, the problem is equivalent to the original problem \eqref{prblm:general}, under a modified type distribution $\tilde{F} = 1 - (1-F)^N$, where $\tilde{F}$ is the distribution of the lowest-cost type of $N$ draws from $F$; thus, our analysis from the baseline applies to this problem.

The design problem focuses on scenarios where the principal can only interact with a single agent; for example, consider a procurement environment where the government is interested in completing a project with a non-monetary objective (such as quality, reduced emissions, or time to completion), the firms have private operational ability, and the government can only contract with a single firm for the project.

The problem in \eqref{prblm:multi} imposes strong incentive constraints; effectivelly, the constraints require that even if all the agents colluded on reporting to maximize agent welfare, the mechanism is still incentive compatible. As such, the value of the problem from \eqref{prblm:multi} provides a lower bound on the value of the multi-agent problem with weaker incentive constraints on the agents. An alternative approach might be to consider Bayesian incentive constraints; that is, agents find it incentive compatible to report truthfully, given a prior over the types of the other agents. Mechanisms with these weakened incentive constraints can look quite strange; we construct a simple example where the optimal mechanism under these constraints is even nonmonotone.

Consider a three-type scenario with two agents; each agent's cost function is $\theta x^2$, where $\theta \in \Theta = \{1, 2, 3\}$. Fix the budget $T = 3$, and consider a scenario where the type is drawn $\theta = 3$ with probability $1-\epsilon$, $\theta = 2$ with probability $\delta \epsilon$, and $\theta = 1$ with probability $\epsilon(1-\delta)$, for $\epsilon, \delta > 0$. Intuitively, when $\epsilon$ is small, the agents are most likely $\theta = 3$, so IR binds to extract as much as possible from the most common type: hence $x(3) = 1$, and $t(3) = 3$. However, conditional on an agent not being type 3, the agent is most likely to be type 1 when $\delta$ is small. In a monotone schedule, the agent of type 2 receives information rent, because type 2 could report being type 3. However, the presence of a second agent means that the designer could dramatically decrease the information rents type $\theta=2$ receives by using the fact that type $\theta=2$ is more likely to win (i.e. be contracted with) than type 1; as a result, the designer can pay type $\theta=2$ less than the full budget (and potentially ask less than $x = 1$) and still satisfy Bayesian incentive compatibility. This has a benefit to the designer when suppressing type $\theta=2$'s information rents means a larger $x$ from type $\theta=1$. Table \ref{tab:nonmonotone} plots numerical values of the optimal BIC mechanism for this example at varying values of $\delta$; indeed, when $\delta$ is small, the optimal mechanism is non-monotone precisely to extract a larger $x(1)$ from the best type. 

The example highlights a key distinction in how an ex-post treatment of budget constraints fundamentally differs from a per-unit cost treatment. In the exercise of a multi-agent problem with BIC constraints, where we replace the budget constraint by inducing a resource cost in the objective of the principal, we find that the optimal mechanisms must be monotone (this can be shown via standard arguments; an envelope theorem can be applied to replace the transfers in the objective, and integration by parts implies that the principal allocates to the best virtual type). However, in the example provided, the reason the principal is willing to sacrifice production of the intermediate type $\theta = 2$ is precisely because the resource is inherently valueless, and so the optimal mechanism actually pays the full budget to both the $\theta = 1$ and $\theta = 3$ types. The optimal mechanism does the former because the type is more efficient, and the latter because the $\theta = 3$ type is more prevalent. Since the $\theta = 2$ type is neither, the optimal mechanism trades off production from that type in order to pay less information rent to the $\theta = 1$ type.

\begin{table}
    \centering
    \begin{tabular}{c|c|c|c}
        $\delta$ & $x(1)$ & $x(2)$ & $x(3)$ \\ \hline
        0.3 & 1.463 & 0.694 & 1 \\
        0.4 & 1.289 & 1.010 & 1 \\
        0.5 & 1.228 & 1.095 & 1 
    \end{tabular}
    \caption{Computed optimal BIC mechanisms with two agents, types in $\{1, 2, 3\}$, budget $T = 3$ and cost function $\theta x^2$. The probability of type 3 is fixed to $0.8$. The probability of type 2 conditional on not being type 3 is $\delta$. Note that for $\delta = 0.3$, the optimal mechanism is non-monotone. }
    \label{tab:nonmonotone}
\end{table}

\subsection{Ex Ante Budget Constraints}
An alternative formulation might require that the budget only needs to bind ex-ante (i.e., the expected transfer cannot exceed $T$). That is, consider the problem replacing the budget constraint (B) with
\begin{align}
&&& \int_{\ubar{\theta}}^{\bar{\theta}} t(\theta) f(\theta) \ d\theta \le T && \forall \theta \in \Theta &&\textnormal{(B')} \notag 
\end{align}
The optimal solution is characterized by the following proposition.
\begin{proposition} \label{prop:exante}
    Replacing (B) with (B'), an optimal mechanism $(x, t)$ exists, and is unique. Further, there exists a nonnegative Lagrange multiplier $\lambda$ and a nonnegative, absolutely continuous costate functions $\rho$ such that $(x,\lambda,\rho)$ is unique solution that satisfies Lemma \ref{lem:general_feasible} and complementary slackness \eqref{cond:comp_slackness} holds, with the modified costate evolution equation
    \begin{equation}\label{cond:costate_exante}
        \dot{\rho}(\theta) = f(\theta)\left(1 - \lambda \left( \Psi_x(x(\theta), \theta) + \Psi_{x\theta}(x(\theta), \theta)\frac{F(\theta)}{f(\theta)} \right) \right) 
    \end{equation}
    upper boundary condition \eqref{cond:upper_bound}, modified lower bound condition $\rho(\ubar{\theta}) = 0$, and modified budget binding condition
    \begin{gather} \label{cond:exante_budget}
         \int f(\theta) \left( \Psi(x(\theta),\theta) + \int_{\theta}^{\bar{\theta}} \Psi_\theta(x(s), s) \ ds\right) \ d\theta  = T 
    \end{gather}
\end{proposition}

Notably, the lower boundary condition is modified to $\rho(\ubar{\theta}) = 0$, and hence the pooling region that appears for the ex-post budget constraint does \textit{not} appear here. Thus, the key feature of our model, the pooling region where the principal trades off between utilizing the budget less or providing information rents, does not emerge in the case where the budget is only required ex-ante.

\section{Conclusion}
\label{sec:conclusion}
 We analyze a mechanism design model where a principal interacts with an agent, where the principal has a budget of a resource it can provide as an incentive. The design problem features an ex-post budget constraint and a first-order objective maximizing the agent's action. We derive the optimal mechanism, and show that the mechanism pools sufficiently efficient types and withholds the budget from less efficient types. The pooling region threshold is crucial to understanding how the mechanism changes when the principal has a resource value; as the resource value increases, the principal demands more or less from agents depending on their type relative to the threshold. 
 
 There are several promising avenues for future research. First, one could extend our framework to dynamic incentives; adding temporal dimensions could help illustrate how optimal mechanisms evolve over time, with changing resource availability and agent behavior. Additionally, our model's applicability to various domains suggests insights for real-world implementations and empirical validations. Lastly, extending the model to allow for more general uncertainty, such as an uncertain resource value or a richer space of cost functions, could enhance the practical relevance of this paper's insights. 

\pagebreak

\bibliographystyle{agsm}
\bibliography{biblio}

\pagebreak

\appendix
\section{Appendix: Supplemental Material}

\subsection{Omitted Proofs}
\NoHyper
\begin{proof}[Proof of Lemma \ref*{lem:general_feasible}]
Consider any $x$ that satisfies the two conditions. Take $t$ given by \eqref{eqn:gen_transfers}. Since $\Psi_\theta$ is nonnegative, $t(\theta) \ge \Psi(x(\theta),\theta)$, so IR is satisfied. Now, we show that $t$ must be nonincreasing. Suppose $\theta \le \theta'$. Then 
{\small \begin{align*}
t(\theta) &= \Psi(x(\theta),\theta) + \int_\theta^{\overline{\theta}} \Psi_\theta(x(s),s) \ ds \\
&= \Psi(x(\theta),\theta) + \int_\theta^{\theta'} \Psi_\theta(x(s),s)\ ds + \int_{\theta'}^{\overline{\theta}} \Psi_\theta(x(s),s)\ ds \\
&\ge \Psi(x(\theta),\theta) + \int_\theta^{\theta'} \Psi_\theta(x(\theta'),s) \ ds + \int_{\theta'}^{\overline{\theta}} \Psi_\theta(x(s),s)\ ds \\
&= \Psi(x(\theta),\theta) + \Psi(x(\theta'),\theta') - \Psi(x(\theta'),\theta)  + \int_{\theta'}^{\overline{\theta}} \Psi_\theta(x(s),s)\ ds \\
&= \left[ \Psi(x(\theta),\theta)- \Psi(x(\theta'),\theta)\right] + t(\theta') 
\end{align*}}
where we used the fact that $x$ is nonincreasing and supermodularity of $\Psi$ in the 3rd step. Also, since $x$ is nonincreasing, and $\Psi$ is increasing in its first argument, the bracketed part is nonnegative, and so we have that $t$ is nonincreasing. Hence, to check the budget constraint, we just have to check that the budget constraint holds at $t(\underline{\theta})$:
\[ t(\underline{\theta}) = \Psi(x(\underline{\theta}),\underline{\theta}) + \int_\Theta \Psi_\theta(x(s),s) \ ds = \int_\Theta \Psi_\theta(x(s),s) \le T \]
Finally, to check that IC is satisfied, consider IC for type $\theta$. We consider the deviation of $\theta$ to $\theta'$. We first show that $\theta$ does not want to deviate up, or $\theta \le \theta'$:
{\small \begin{align*}
t(\theta) - \Psi(x(\theta), \theta) &=  \int_\theta^{\overline{\theta}} \Psi_\theta(x(s), s) \ ds \\ 
&=  \int_\theta^{\theta'} \Psi_\theta(x(s), s) \ ds + \int_{\theta'}^{\overline{\theta}} \Psi_\theta(x(s), s) \ ds \\
&= \int_\theta^{\theta'} \Psi_\theta(x(s), s) \ ds + t(\theta') - \Psi(x(\theta'), \theta') \\
&\ge \int_\theta^{\theta'} \Psi_\theta(x(\theta'), x) + t(\theta') - \Psi(x(\theta'), \theta') \\
&=  \Psi(x(\theta'),\theta') - \Psi(x(\theta'),\theta)+ t(\theta') - \Psi(x(\theta'), \theta') \\
&= t(\theta') - \Psi(x(\theta'),\theta)
\end{align*}}
Now we check that $\theta$ does not want to deviate down to type $\theta' \le \theta$:
\begin{align*}
t(\theta) - \Psi(x(\theta), \theta) &= \int_\theta^{\overline{\theta}} \Psi_\theta(x(s), s) \ ds\\ 
&= -\int_{\theta'}^\theta \Psi_\theta(x(s), s) \ ds + \int_{\theta'}^{\overline{\theta}} \Psi_\theta(x(s), s) \ ds \\
&= -\int_{\theta'}^\theta \Psi_\theta(x(s), s) \ ds + t(\theta') - \Psi(x(\theta'), \theta') \\
&\ge -\int_{\theta'}^\theta  \Psi_\theta(x(\theta'), s) \ ds + t(\theta') - \Psi(x(\theta'), \theta') \\
&= - \Psi(x(\theta'),\theta) + \Psi(x(\theta'),\theta')+ t(\theta') - \Psi(x(\theta'), \theta') \\
&= t(\theta') - \Psi(x(\theta'),\theta)
\end{align*}
So IC is satisfied. Thus, if $x$ satisfies the conditions, it is feasible.

Now, we show the two conditions are necessary. Consider any feasible $x: \Theta \to \mathbb{R}_+$. Suppose $\theta\le \theta'$, and $\theta, \theta' \in \Theta$. From rewriting the IC constraints, we get that 
    \[ \Pi(x(\theta); \theta) - \Pi(x(\theta'); \theta) \ge t(\theta') - t(\theta) \]
    \[ \Pi(x(\theta'); \theta') - \Pi(x(\theta); \theta') \ge t(\theta) - t(\theta') \]\
Note that the $\Pi_{LF}(\theta)$ and $\Pi_{LF}(\theta')$ terms will cancel on the left hand sides. Adding these together, we get 
\[ -\Psi(x(\theta), \theta) + \Psi(x(\theta'), \theta) - \Psi(x(\theta'), \theta') + \Psi(x(\theta),\theta') \ge 0  \]
Rearranging
\[ \Psi(x(\theta'), \theta) +  \Psi(x(\theta),\theta') \ge \Psi(x(\theta), \theta) + \Psi(x(\theta'), \theta')   \]
And so supermodularity of $\Psi$ implies that $x(\theta') \le x(\theta)$, and so the first condition must hold. To show that the transfers must take the form given, define the interim utility
\[ U(\theta) = t(\theta) - \Psi(x(\theta),\theta) \]
IC implies that 
\[ U(\theta) = \max_{\theta'} t(\theta') - \Psi(x(\theta'),\theta) \]
By the envelope theorem,
\[ U'(\theta) = - \Psi_\theta(x(\theta),\theta) \]
\[ U(\theta) = U(\bar{\theta}) + \int_\theta^{\bar{\theta}} \Psi_\theta(x(s),s) \ ds \]
So transfers must be of the form 
\begin{align*}
    t(\theta) &= U(\theta)  \\
    &= U(\bar{\theta})+ \Psi(x(\theta), \theta) + \int_\theta^{\bar{\theta}} \Psi_\theta(x(s),s) \ ds
\end{align*} 
In particular, the transfer for type $\ubar{\theta}$ satisfies 
\begin{align*}
    T &\ge t(\ubar{\theta}) \\
    &=  U(\bar{\theta})+ \Psi(x(\ubar{\theta}),\ubar{\theta}) + \int_{\ubar{\theta}}^{\bar{\theta}} \Psi_\theta(x(s), s) \ ds \\
    &\ge  \Psi(x(\ubar{\theta}),\ubar{\theta}) + \int_{\ubar{\theta}}^{\bar{\theta}} \Psi_\theta(x(s), s) \ ds 
\end{align*} 
since $U(\bar{\theta}) \ge 0$ by IR. Hence, the two conditions are necessary and sufficient for a feasible schedule.
\end{proof}

\begin{proof}[Proof of Theorem \ref*{thm:general_mech}]
By Lemma \ref*{lem:general_feasible}, we can rewrite the design problem as 
\begin{align}
& \max_{x} && \int_{ \Theta} x(\theta) f(\theta) \ d\theta \label{eqn:general_prob_reduced} \\
& \textnormal{subject to} && \Psi(x(\ubar{\theta}),\ubar{\theta}) + \int_{\ubar{\theta}}^{\bar{\theta}} \Psi_\theta(x(s), s) \ ds  \le T &&&& \textnormal{(Normalization)} \notag  \\
&&& x \textnormal{ nonincreasing} &&&&\textnormal{(Monotonicity)} \notag
\end{align}
The Lagrangian relaxation is
\begin{align*}
& \max_{x} && \int_{ \Theta} x(\theta) f(\theta) \ d\theta + \lambda\left(T - \Psi(x(\ubar{\theta}),\ubar{\theta}) - \int_{\ubar{\theta}}^{\bar{\theta}} \Psi_\theta(x(s), s) \ ds  \right)  \\
& \textnormal{subject to} && x \textnormal{ nonincreasing} &&&&\textnormal{(Monotonicity)} 
\end{align*}
By monotonicity, $x$ must be differentiable almost everywhere. The monotonicity constraint then requires that wherever $x$ is differentiable, $\dot{x} \le 0$. Letting $u = \dot{x}$ and rewriting the objective function, the problem becomes 
\begin{align}
& \max_{x} && \int_{ \Theta} \left( x(\theta) f(\theta) - \lambda \Psi_\theta(x(\theta), \theta) \right)  \ d\theta + \lambda T - \lambda \Psi(x(\ubar{\theta}),\ubar{\theta})  
 \label{ctrl:lagrangian} \\
& \textnormal{subject to} && \dot{x} = u \le 0 &&&&\textnormal{(Monotonicity)}  \notag
\end{align}
The problem now is an optimal control problem with a scrap value constraint; $u$ is the control variable, and $x$ is the state variable. To argue that this rewriting is without loss, we adopt the \cite{lt86} argument. The original problem optimizes over the space of nonincreasing $x$. Given an optimal solution $x^*$ to \eqref{ctrl:lagrangian} that is absolutely continuous and nonincreasing, $x^*$ must also be optimal over the original space, since the space of absolutely continuous nonincreasing functions is dense in the space of nonincreasing functions under the weak norm topology, and the maximization objective is continuous.

We first show that an optimal mechanism $(x,t)$ satisfies the conditions \eqref{cond:comp_slackness}-\eqref{cond:budget_bind}. For a fixed $\lambda > 0$, Pontryagin's maximum principle implies that any optimal solution admits a nonnegative, absolutely continuous control variable such that \eqref{cond:comp_slackness} - \eqref{cond:lower_bound} are true. Further, note that if $\lambda = 0$, \eqref{cond:lower_bound} implies that $\rho(\ubar{\theta}) = 0$, but \eqref{cond:costate_evol} implies that $\dot{\rho}(\ubar{\theta}) = - f(\ubar{\theta}) < 0$, which is impossible since $\rho$ must be nonnegative. Hence, $\lambda > 0$, so by Lagrangian duality we must have that the budget binds, so \eqref{cond:budget_bind} holds. 

Second, we show that any solution that satisfies \eqref{cond:comp_slackness}-\eqref{cond:budget_bind} must be optimal. Note that fixing $\lambda$, the maximized Hamiltonian corresponding to the control problem is concave in the state variable $x$, so the Arrow sufficiency condition holds.\footnote{See \cite{ks71} for details.} Therefore, the solution maximizes the Lagrangian and and the budget constraint \eqref{cond:budget_bind} binds, the solution must be optimal in the original design problem. 

Finally, we argue that an optimal mechanism exists (uniqueness holds from Corollary \ref{corr:unique_soln}). To show an optimal mechanism exists, we show that a solution to conditions \eqref{cond:comp_slackness}-\eqref{cond:budget_bind} exists. 
Fixing $\lambda$, the conditions \eqref{cond:comp_slackness}-\eqref{cond:lower_bound} admits a solution\footnote{More precisely, a solution in the extended sense, i.e. absolutely continuous and satisfying the differential equations almost everywhere.} $(x_\lambda, \rho_\lambda)$ by the Caratheodory existence theorem. Define $t_\lambda$ as the transfer schedule for $x_\lambda$ by equation \eqref{eqn:gen_transfers}. We now argue that a $\lambda$ exists such that \eqref{cond:budget_bind} holds. To show such a $\lambda$ exists, we argue that $t_\lambda(\ubar{\theta})$ is continuous in $\lambda$, $t_\lambda(\ubar{\theta}) \to 0$ for $\lambda \to \infty$, and $t_\lambda(\ubar{\theta}) \to \infty$ for $\lambda \to 0$; the intermediate value theorem then implies the result. 

For continuity of $t_\lambda(\ubar{\theta})$, note that taking any sequence $\{ \lambda_k \} \to \lambda$, the family of control problems \eqref{ctrl:lagrangian} admits an optimal $x_{\lambda_k}$, which must converge $x_{\lambda_k} \to x_\lambda$ weakly in $H^1$ by Theorem 3.1 from \cite{w01}. Hence, by construction of $t_\lambda$, $t_{\lambda_k}(\ubar{\theta}) \to t_\lambda(\ubar{\theta})$. 

To characterize the boundary, note that as $\lambda \to \infty$, the limit control problem is simply 
\begin{align*}
& \max_{x} && \lambda \left( \int_{ \Theta} -\Psi_\theta(x(\theta), \theta)  \ d\theta + T - \Psi(x(\ubar{\theta}),\ubar{\theta})  \right) + o(\lambda)
\\
& \textnormal{subject to} && \dot{x} = u \le 0 &&&&\textnormal{(Monotonicity)} 
\end{align*}
which has an optimal solution of $x = 0$ everywhere in the limit $\lambda \to \infty$, and hence $t_\lambda(\ubar{\theta}) \to 0$ as $\lambda \to \infty$. Finally, note that as $\lambda \to 0$, the problem becomes 
\begin{align*}
& \max_{x} &&  \int_{ \Theta} x(\theta)f(\theta) \ d\theta
\\
& \textnormal{subject to} && \dot{x} = u \le 0 &&&&\textnormal{(Monotonicity)} 
\end{align*}
which is unbounded, and hence $t_\lambda(\ubar{\theta}) \to \infty$ as $\lambda \to 0$. Thus, by the intermediate value theorem, there exists a $\lambda$ such that $t_\lambda(\ubar{\theta}) = T$. 
\end{proof}

\begin{proof}[Proof of Proposition \ref{prop:separable}]
As in the proof of Theorem \ref{thm:general_mech}, we write out the Lagrangian:
\begin{align*}
& \max_{x} && \int_{ \Theta} x(\theta) f(\theta) \ d\theta + \lambda\left(T - \ubar{\theta}\Gamma(x(\ubar{\theta})) - \int_{\ubar{\theta}}^{\bar{\theta}} \Gamma(x(s)) \ ds  \right)  \\
& \textnormal{subject to} && x \textnormal{ nonincreasing} &&&&\textnormal{(Monotonicity)} 
\end{align*}
Define $f(\theta) = 0$ for $\theta < \ubar{\theta}$. Then note that the optimization can be rewritten equivalently as 
\begin{align*}
& \max_{x} && \int_{0}^{\bar{\theta}} x(\theta) f(\theta) \ d\theta + \lambda\left(T - \int_{0}^{\bar{\theta}} \Gamma(x(s)) \ ds  \right)  \\
& \textnormal{subject to} && x \textnormal{ nonincreasing} &&&&\textnormal{(Monotonicity)} 
\end{align*}

To handle the monotonicity constraint, we use a generalized ironing technique. By Theorem 4.4 in \cite{t11}, the optimal $x(\theta)$ then maximizes $\overline{J}(\cdot, \theta; \lambda)$ at the optimal $\lambda$, where 
\[ \overline{J}(x, \theta; \lambda) = \int_{0}^s \frac{\partial}{\partial \theta^-}\textnormal{cav}_\theta\left[ F(\theta) - \lambda\theta \Gamma'(x) \right] d\theta \]
and $\textnormal{cav}_\theta$ denote the concavification operation in $\theta$, where the differentiation is taken from the left (as the concavification need not be smooth in $\theta$). Note that since $(\theta - \ubar{\theta})$ is linear in $\theta$, we can rewrite $\overline{J}$ as 
\[ \overline{J}(x, \theta; \lambda) = \int_{0}^x \frac{\partial}{\partial \theta^-}\textnormal{cav}_\theta\left[ F(\theta) \right]d\theta - \lambda \Gamma'(x)  \]
\[ = \tilde{f}(\theta) -  \lambda \Gamma'(x)  \]
Hence, the problem is reduced to a point-wise optimization at each $\theta$: inverting $\Gamma'$ implies that 
\[ x^*(\theta)= (\Gamma')^{-1}\left(\frac{\tilde{f}(\theta)}{\lambda}\right) \]
It remains to argue that a $\lambda$ exists such that the budget constraint binds. Note that as $\lambda \to 0$, $x^* \to \infty$, and as $\lambda \to \infty$, $x^* \to 0$. Hence, the intermediate value theorem implies some finite value of $\lambda$ makes the budget constraint bind.
\end{proof}
\begin{proof}[Proof of Proposition \ref{prop:linear_schedule_bad}]
We construct a transfer schedule $t$ that implements $x_r$ and then show that that transfer schedule does not exhaust the budget, implying that we could increase production on some measure of types. Let $\hat{\theta}_r$ be the largest $\theta$ such that $\Psi_x(T/r, \theta) \le r$; that is, the largest $\theta$ whose marginal cost of producing at $T/r$ is less than $r$.

First, suppose that $\hat{\theta}_r = \bar{\theta}$. Then under $r$, any type of firm produces $T/r$. So 
    \begin{align*}
    \Psi(x_r(\ubar{\theta})\ubar{\theta}) + \int_{\ubar{\theta}}^{\bar{\theta}} \Psi_\theta(x_r(\theta), \theta) \ d\theta = \Psi(T/r, \bar{\theta}) < \frac{T}{r}\Psi_x(T/r, \hat{\theta}_r) \le \frac{T}{r}r = T
    \end{align*}
    where the first inequality follows from convexity of $\Psi$ in $x$, and the third line follows from the definition of $\hat{\theta}_r$. Thus, the transfer schedule given by \eqref{eqn:gen_transfers} implements $x_r$ without exhausting the budget.

    Now, suppose $\hat{\theta}_r < \bar{\theta}$. Then 
    \begin{align}
    \Psi(x_r(\ubar{\theta})\ubar{\theta}) + \int_{\ubar{\theta}}^{\bar{\theta}} \Psi_\theta(x_r(\theta), \theta) \ d\theta =& \int_{\underline{\theta}}^{\hat{\theta}_r} \Psi_\theta (T/r, \hat{\theta}_r) d\theta + \int_{\hat{\theta}_r} \Psi_\theta(x_r(\theta), \theta) d \theta \notag \\
    =& \Psi(T/r, \hat{\theta}_r) +\int_{\hat{\theta}_r} \left[ \Psi_\theta(x_r(\theta), \theta) + \Psi_x(x_r(\theta), \theta) x_r'(\theta)\right] d \theta \notag \\& - \int_{\hat{\theta}_r} \Psi_x(x_r(\theta), \theta) x_r'(\theta) d\theta \notag \\
    =& \Psi(T/r, \hat{\theta}_r) +\left[ \Psi(x_r(\overline{\theta}), \overline{\theta})-\Psi(x_r(\hat{\theta}_r), \hat{\theta}_r)\right] - \int_{\hat{\theta}_r} \Psi_x(x_r(\theta), \theta) x_r'(\theta) d\theta \notag\\
    =& \Psi(x_r(\overline{\theta}), \overline{\theta})- r \int_{\hat{\theta}_r} x_r'(\theta) d\theta \label{eqn:linear_step_1}\\
    =& \Psi(x_r(\overline{\theta}), \overline{\theta}) - r \left(x_r(\overline{\theta}) - x_r(\hat{\theta}_r)\right) \notag \\
    =& \Psi(x_r(\overline{\theta}), \overline{\theta}) - r \left(x_r(\overline{\theta}) - T/r\right) \notag\\
    =& T + \Psi(x_r(\overline{\theta}), \overline{\theta}) - \Psi_x(x_r(\overline{\theta}),\overline{\theta})x_r(\overline{\theta}) < T \notag
    \end{align}
    where \eqref{eqn:linear_step_1} uses the fact that the agent equalizes marginal cost with the subsidy rate, $\Psi_x(x, \theta) = r$.
\end{proof}

\begin{proof}[Proof of Theorem \ref{thm:linear_cost}]
The analysis proceeds exactly as in Theorem \ref{thm:general_mech}; in this case, we plug in the transfers from Lemma \ref{lem:general_feasible} into the objective and integrate by parts. The Hamiltonian becomes 
\[ \mathcal{H} = x f(\theta)  - \lambda \Psi_\theta(x, \theta) - k \Psi(x,\theta) f(\theta) - k \Psi_\theta(x, \theta) F(\theta) + \rho u \]
Applying Pontryagin, the costate evolution \eqref{cond:costate_linear} follows.
\end{proof}

\begin{proof}[Proof of Proposition \ref{prop:linear_comparative}]
    Since ${\Psi_{x\theta}(x, \theta)}/{f(\theta)}$ is weakly increasing in $\theta$, it is easy to check that the $x$ that solves $\dot{\rho} = 0$ is nonincreasing. 
    First, define the subproblem on $[\hat{\theta}, \bar{\theta}]$ with no initial conditions:
    \begin{align}
& \max_{x, t} && \int_{\hat\theta}^{\bar\theta} \left[  x(\theta) - k \left(\Psi(x(\theta), \theta) +  \int_{\theta}^{\bar{\theta}} \Psi_\theta(x(s), s) \ ds \right) - \lambda \Psi_\theta(x(\theta),\theta)\right]f(\theta) \ d\theta \label{prblm:linear_subproblem} \\
& \textnormal{subject to} && x \textnormal{ nonincreasing} &&&&\tag{Monotonicity} 
\end{align}
    Note that any solution to the conditions in Theorem \ref{thm:linear_cost} restricted to $[\hat{\theta}, \bar{\theta}]$ also solve the Pontryagin necessary conditions for this subproblem, and hence by Arrow's sufficiency theorem, the subsolution must also optimize this subproblem. 
    
    Then by Proposition \ref{prop:linear}, the objects $\bar{x}$, the threshold type $\hat{\theta}$, and $\lambda$ jointly solve the following three equations: 
    \begin{align*}
        G_1(\bar{x},\hat{\theta},\lambda; k) &= \Psi(\bar{x},\hat{\theta}) + \int_{\hat{\theta}}^{\bar{\theta}}\Psi_\theta(x(s),s) \ ds - T &= 0 \\
        G_2(\bar{x},\hat{\theta},\lambda; k) &=\lambda \Psi_x(\bar{x},\hat{\theta}) - F(\hat{\theta})(1 - k \Psi_x(\bar{x}, \hat{\theta}))  &= 0 \\
        G_3(\bar{x},\hat{\theta},\lambda; k) &= \lambda \Psi_x(\bar{x},\hat{\theta}) +  k( \Psi_x(\bar{x}, \hat{\theta}) f(\hat{\theta}) + \Psi_{x\theta}(\bar{x}, \hat{\theta})F(\hat{\theta})) - f(\hat{\theta}) &= 0 
    \end{align*}
    where $x(s)$ solves the point-wise first-order condition from the maximized Hamiltonian of the subproblem \eqref{prblm:linear_subproblem} at the optimal $u$:
    \[ H'(x(s),s,\lambda;k) = \lambda \Psi_x(x(s),s) +  k( \Psi_x(x(s), s) f(s) + \Psi_{x\theta}(x(s), s)F(s)) - f(s) = 0 \]
    We apply the implicit function theorem to find the derivatives. The Jacobean of the system evaluated at $(\bar{x},\hat{\theta},\lambda)$ is 
    \[ J = \begin{bmatrix}
        \Psi_x & 0 & \int_{\hat{\theta}}^{\bar{\theta}} \Psi_{x\theta}(x(s),s) \frac{dx(s)}{d\lambda} \ ds \\
        (\lambda + k F)\Psi_{xx} & 0 & \Psi_x \\
        k(\Psi_{xx} f + \Psi_{xx\theta}F) + \lambda \Psi_{xx\theta} & k(2\Psi_{x\theta}f + \Psi_x f' + \Psi_{x\theta\theta}F) + \lambda \Psi_{x\theta\theta} - f' &  \Psi_{x\theta}   
        \end{bmatrix} \]
        where the functions are understood to be evaluated at $(\bar{x},\hat{\theta},\lambda)$ where appropriate. The Jacobean is invertible since the partial $\partial G_1/\partial \lambda$ is negative because $dx(s)/d\lambda$ is negative, and the other nonzero entries are positive. Note that 
        \[ \frac{\partial G_1}{\partial k} = \int_{\hat{\theta}}^{\bar{\theta}} \Psi_{x\theta}(x(s),s) \frac{dx(s)}{dk} \ ds < 0 \]
        \[ \frac{\partial G_2}{\partial k} = \Psi_x F \]
        Then from the implicit function theorem, we have 
        \[ \frac{d\bar{x}}{dk} = \frac{\Psi_x \int_{\hat{\theta}}^{\bar{\theta}} \Psi_{x\theta}(x(s),s) \left(  \frac{dx(s)}{d\lambda}F(\hat{\theta}) -\frac{dx(s)}{dk}  \right) \ ds}{(\Psi_x)^2 - (\lambda + kF)\Psi_{xx} \int_{\hat{\theta}}^{\bar{\theta}} \Psi_{x\theta}(x(s),s) \frac{dx(s)}{d\lambda} \ ds} \]
        \[ \frac{d\lambda }{dk} = \frac{ (\lambda + kF)\Psi_{xx} F\int_{\hat{\theta}}^{\bar{\theta}} \Psi_{x\theta}(x(s),s) \frac{dx(s)}{dk}\ ds - (\Psi_x)^2 F}{(\Psi_x)^2 - (\lambda + kF)\Psi_{xx} \int_{\hat{\theta}}^{\bar{\theta}} \Psi_{x\theta}(x(s),s) \frac{dx(s)}{d\lambda} \ ds} \]
        From the observation that $dx(s)/d\lambda < 0$, it follows that the denominators are strictly positive, so the sign of each of these derivatives is the same as the sign of the numerator. It follows then that $d\lambda / dk < 0$, since $dx(s)/dk$ is negative. To determine the sign of $d\bar{x}/dk$, we apply the implicit function theorem for $dx(s)/d\lambda$ and $dx(s)/dk$, and obtain that:
        \begin{align*} &\int_{\hat{\theta}}^{\bar{\theta}} \Psi_{x\theta}(x(s),s) \left(  \frac{dx(s)}{d\lambda}F(\hat{\theta}) -\frac{dx(s)}{dk}  \right) \ ds \\
        &= \int_{\hat{\theta}}^{\bar{\theta}} \Psi_{x\theta}(x(s),s) \left(  \frac{- \frac{\partial H}{\partial \lambda }}{\frac{\partial H}{\partial x}}F(\hat{\theta}) - \frac{-\frac{\partial H}{\partial k}}{\frac{\partial H}{\partial x}}  \right) \ ds \\
        &= \int_{\hat{\theta}}^{\bar{\theta}} \Psi_{x\theta}(x(s),s) \left(  \frac{- \Psi_x(x(s),s)}{\frac{\partial H}{\partial x}}F(\hat{\theta}) + \frac{\Psi_x(x(s), s) f(s) + \Psi_{x\theta}(x(s), s)F(s)}{\frac{\partial H}{\partial x}}  \right) \ ds\\
        &= \int_{\hat{\theta}}^{\bar{\theta}} \Psi_{x\theta}(x(s),s) \left(  \frac{\Psi_x(x(s), s) f(s) + \Psi_{x\theta}(x(s), s)(F(s) - F(\hat{\theta}))}{\frac{\partial H}{\partial x}}  \right) \ ds
        \end{align*}
        Since $\partial H /\partial x$ is positive, and $s \ge \hat{\theta}$ it follows that this quantity is positive; hence $d\bar{x}/dk$ is positive.
\end{proof}

\begin{proof}[Proof of Proposition \ref{prop:exante}]
    By Lemma \ref{lem:general_feasible}, we can reduce (IC) and (IR) into a monotonicity constraint and a transfer schedule. Adding a Lagrange multiplier $\lambda$ to the budget constraint, the control problem becomes
    \begin{align*}
& \max_{x} && \int_{ \Theta} x(\theta) f(\theta)  \ d\theta + \lambda \left( T -  \int_{\ubar{\theta}}^{\bar{\theta}} \left(\Psi(x(\theta), \theta) +  \int_{\theta}^{\bar{\theta}} \Psi_\theta(x(s), s) \ ds  \right) f(\theta) \ d\theta \right)    \\
&&& \dot{x} = u \le 0 
\end{align*}
Integrating by parts, and writing the integral in one term, we get
\begin{align*}
& \max_{x} && \int_{ \Theta} \left( x(\theta)- \lambda \left( \Psi(x(\theta), \theta) + \Psi_\theta(x(\theta), \theta)\frac{F(\theta)}{f(\theta)} \right) \right)  f(\theta)  \ d\theta + \lambda T   \\
&&& \dot{x} = u \le 0 
\end{align*}
Applying Pontryagin and the Arrow sufficiency result, we get the conditions in the proposition. 
\end{proof}

\end{document}